\documentclass{article}
\usepackage{graphicx} 
\usepackage[utf8]{inputenc}
\usepackage{braket}
\usepackage{amsfonts}
\usepackage{amsmath}
\usepackage{amssymb}
\usepackage{amsthm}
\usepackage{mathrsfs}
\usepackage[a4paper,total={6in,9in}]{geometry}
\usepackage{setspace}
\usepackage{verbatim}
\usepackage{dsfont}
\usepackage{tikz}
\usepackage[dvipsnames]{xcolor}
\usetikzlibrary{arrows.meta}
\usetikzlibrary{calc}
\usetikzlibrary{patterns}
\usetikzlibrary{3d}
\usepackage{comment}
\usepackage{centernot}
\usepackage{arydshln}
\usepackage{mathtools}
\usepackage{afterpage}

\usepackage{appendix}
\usepackage{url}
\usepackage{bm}
\usepackage{pgfplots}
\pgfplotsset{compat=1.9}
\usepgfplotslibrary{groupplots}

\usepackage[style=numeric, sorting=none]{biblatex} 
\newtheorem{theorem}{Theorem}[section]
\newtheorem{proposition}[theorem]{Proposition}
\newtheorem{lemma}[theorem]{Lemma}

\newtheorem{definition}[theorem]{Definition}

\newcommand{\ind}[1]{\mathds{1}_{\{#1\}}}
\newcommand{\indnb}[1]{\mathds{1}_{#1}}

\newcommand{\nlr}{\not\leftrightarrow}

\newcommand{\texpec}{\tilde{\mathds{E}}}

\newcommand\blfootnote[1]{%
  \begingroup
  \renewcommand\thefootnote{}\footnote{#1}%
  \addtocounter{footnote}{-1}%
  \endgroup
}

\title{Cluster Representation of Renormalization Group Transformations and a Rigorous Proof for Convergence of the RG-Flow of the Ising Model to Trivial Fixed Points away from Criticality}
\bigskip

\author{Fabio Arz}

\date{Albert Einstein Center for Fundamental Physics\\
Institute for Theoretical Physics\\
University of Bern, Sidlerstrasse 5, CH-3012 Bern, Switzerland\\[2ex]
    \today}

\begin{document}

\maketitle

\begin{abstract}
    Many rigorous results in the modern era of statistical mechanics have been obtained through geometrical representations. A much studied tool in this setting is the random cluster representation of lattice spin models. Another area of statistical mechanics that is of great interest but lacking rigorous results is the theory of the renormalization group. This paper investigates the idea to find a common ground between these two concepts in order to obtain rigorous results on the renormalization group flow. A need for negative cluster-weights weakens the success of this approach. Nevertheless, we managed to establish a relation between the scaling limit of the renormalization group flow of the nearest-neighbour Ising model away from criticality with a simple one-dimensional dynamical system that follows the cluster connectivity of the renormalization group transformation. This allows for a rigorous proof of the convergence of this flow to the zero- and infinite-temperature fixed point respectively for a large family of renormalization group transformations. Explicit results will be established on $\mathds{Z}^2$ followed by a discussion of the available generalizations to higher dimensions.
\end{abstract}

\blfootnote{*E-mail: fabio.arz@unibe.ch}

\newpage

\tableofcontents

\newpage

\section{Introduction}

The Ising model was first studied in 1925 \cite{Ising} as a simple mathematical model to study phase transitions. Although it received relatively little initial attention, it has since become one of the most extensively studied and thoroughly understood models in statistical physics. This enduring prominence stems from its ability to exhibit rich phase transition behaviour despite its mathematically straightforward formulation.

The renormalization group (RG) framework emerged much later. In the late 1940s, physicists proposed renormalization techniques to tame the infinities plaguing the rising field of quantum field theory. Attempts to formalize these techniques led to early versions of the renormalization group. In 1966, Kadanoff \cite{Kadanoff} introduced the intuitive heuristic of block-spin transformations to investigate the $2\text{D}$ Ising model close to the critical temperature. By averaging out short-range degrees of freedom into collective block spins, Kadanoff showed how a system could be viewed at progressively larger length scales. In the 1970s, Wilson \cite{WilsonI},\cite{WilsonII} combined the renormalization group idea with Kadanoff's block-spin transformations into a definitive mathematical framework. Under Wilson’s paradigm, this iterative coarse-graining defines a flow of the system's Hamilton function through a space of coupling constants toward a scale-invariant fixed point.

Crucially, this new framework provided a physical explanation for the principle of universality, the phenomenon where large-scale physics depends only weakly on the exact details of microscopic interactions. The RG-flow systematically filters out these short-range details, preserving only the structural information relevant to a given scale. In particular, near the critical point of a system like the Ising model, long-range fluctuations dominate the physics. As a result, diverse systems whose trajectories flow toward the same fixed point share identical macroscopic critical behaviour, placing them within the same universality class. Consequently, this principle underscores the profound utility of the Ising model: by investigating its relatively simple interactions, it is possible to deduce the macroscopic behaviour of far more complex systems within the same universality class.

While significant progress has been made in certain areas of the field, a complete mathematical foundation is still lacking. For example, in a major breakthrough, Chelkak and Smirnov \cite{Chelkak-Smirnov} proved the conformal invariance of the $2\text{D}$ Ising model in the continuum limit, confirming a key prediction of RG-theory. However, their proof relies on specific, exact identities of the Ising model and thus fails to account for the broader mechanism of universality. Beyond the challenge of proving universality, the RG framework also suffers from internal mathematical issues. For instance, van Enter, Fernández and Sokal \cite{van_Enter_1993} rigorously demonstrated the occurrence of RG pathologies that strictly limit the validity of Wilson's real-space picture. Ultimately, establishing a rigorous mathematical basis for the principles driving RG flows remains a highly desirable goal, as it would resolve many outstanding problems in mathematical physics.\smallskip

This article is structured as follows. First, we want to establish a common notation required to formulate the main results. Afterwards, the remainder of the introduction is spent on a brief introduction to the random cluster model and its connection to the nearest-neighbour Ising model as well as some measure-theoretic technicalities. In Section \ref{Sec2} we will introduce the general concept of representing RGTs as a model of connectivity. In Section \ref{Sec3} we will then quickly go through decimation as an introductory example before we turn our attention to the family of $2\times2\to1$ block-spin transformations on the square lattice in Section \ref{Sec4}. There we will prove convergence of the RG-flow to the zero- or infinite temperature fixed points for all nearest-neighbour Ising models except the critical point for a suitable set of transformation, while also showcasing pathological behaviour for the remainder of the $2\times 2\to1$ maps. At last, Section \ref{Sec5} contains a short discussion of whether these results generalize to i) higher dimensions and ii) higher state Potts models.

\subsection{The Nearest-Neighbour Ising Model on $\mathds{Z}^2$}

This section presents the most fundamental results about the nearest-neighbour Ising model on $\mathds{Z}^2$. For a more complete picture, a well-written overview by Duminil-Copin of this and related topics including proofs can be found in \cite{10.1007/978-3-030-32011-9_2}.

\begin{definition}
    A graph is a pairing $G=(V,E)$ where $V$ is some set of vertices and $E\subseteq\{\{x,y\},x,y\in V\}$ is a set of edges.
\end{definition}

We call two vertices adjacent, denoted as $x\sim y$, in case that $\{x,y\}\in E$. Further, in that case we may denote the corresponding edge also as $\braket{xy}$. Let us also collect some useful notions:
\begin{enumerate}
    \item[i)] A subgraph $\Gamma\subseteq G$ is a graph $\Gamma=(V_\Gamma,E_\Gamma)$ such that $V_\Gamma\subseteq V_G$ and $\{x,y\}\in E_\Gamma\implies \{x,y\}\in E_G$. We call $\Gamma$ an induced subgraph in case that
    \begin{equation}
        \{x,y\}\in E_\Gamma\iff\{x,y\}\in E_G\quad\forall x,y\in V_\Gamma.
    \end{equation}
    The boundary of a subgraph $\Gamma\subset G$ is defined as $\partial\Gamma:=\{x\in V_\Gamma\,:\,\exists y\in V_G,\braket{xy}\notin E_\Gamma\}$.
    \item[ii)] A path is a sequence $x_i,i=1,...,n$ such that $x_i\sim x_{i+1}\,\forall i\in\{1,...,n-1\}$. We call a graph connected $\iff$ for every pair of vertices $x,y\in V$ there exists a path with endpoints $x,y$.
\end{enumerate}
Let us now zone in onto the standard $2\text{D}$ square lattice, which is a graph with vertex set $\mathcal{V}=\mathds{Z}^2$ and edge set $\mathcal{E}=\{\{x,y\}\,:\,\|x-y\|=1\}$. Further, we choose some finite\footnote{Here, finite simply means finite set of vertices $|V|<\infty$.} subgraph $G=(V,E)\subset\mathds{Z}^2$. For simplicity, unless explicitly stated otherwise, we will always choose any subgraph of $\mathds{Z}^2$ to be connected, induced and such that the complement of $G$ is connected as well. We then assign a configuration space
\begin{equation}
    \Sigma_G:=\{-1,1\}^V.
\end{equation}
This is all we need in order to introduce the nearest-neighbour Ising model with
\begin{enumerate}
    \item[i)]Free boundary conditions:
    \begin{equation}
        \mu_{G,\beta}^\mathrm{f}[\mathcal{A}]=\frac{\sum_{\sigma\in\Sigma_G}\indnb{\mathcal{A}}(\sigma)e^{-\beta H_G^\mathrm{f}(\sigma)}}{\sum_{\sigma\in \Sigma_G}e^{-\beta H_G^\mathrm{f}(\sigma)}},\quad H_G^\mathrm{f}(\sigma):=-\sum_{\braket{xy}\in E}\sigma_x\sigma_y,
    \end{equation}
    \item[ii)] With $\pm$ boundary conditions: 
    \begin{equation}
        \mu_{G,\beta}^\pm[\mathcal{A}]=\mu_{G,\beta}^\mathrm{f}[\mathcal{A}\,|\,\sigma_x=\pm1,\forall x\in\partial G].
    \end{equation}
\end{enumerate}
We denote $G_n\uparrow\mathds{Z}^2$ for a sequence of subgraphs $G_n=(V_n,E_n)\subset\mathds{Z}^2,n\in\mathds{N}$ in case that
\begin{enumerate}
    \item[i)] $G_n\subset G_{n+1}$,
    \item[ii)] $\bigcup_{n\in\mathds{N}} V_n=\mathcal{V}$.
\end{enumerate}
One can prove that for any sequence $G_n\uparrow\mathds{Z}^2$ the measures $\mu_{G_n,\beta}^\mathrm{f},\mu_{G_n,\beta}^\pm$ converge weakly to measures $\mu_\beta^\mathrm{f},\mu_\beta^\pm$ in the limit $n\to\infty$, independent of the exact choice of $G_n$. See for example \cite{friedli_velenik_2017}. Of course, we could also consider more general boundary conditions $\mu_{G,\beta}^\tau[\cdot]:=\mu_{G,\beta}^\mathrm{f}[\cdot\,|\,\sigma_x=\tau_x,\forall x\in\partial G]$ for some $\tau\in\Sigma_G$. However, a series of results in the 1970s proved that on $\mathds{Z}^2$ any infinite volume Gibbs measure compatible with the nearest-neighbour interaction at inverse temperature $\beta$ is translation invariant and can be written as a convex combination of $\mu_\beta^-$ and $\mu_\beta^+$. See \cite{CoVe2012} for an exposure to this matter. Also, in case one is unfamiliar with the meaning behind a measure being compatible with a well-behaved interaction, one can find a very detailed introduction in the book by Friedli and Velenik \cite{friedli_velenik_2017} in Chapter 6. The Ising model on $\mathds{Z}^2$ undergoes a continuous phase transition at a finite inverse temperature $\beta_c=\frac{1}{2}\log(1+\sqrt{2})$. At inverse temperatures $\beta\leq\beta_c$ there is a unique infinite volume Gibbs measure and thus we find that $\mu_\beta^-=\mu_\beta^\mathrm{f}=\mu_\beta^+$. At inverse temperatures $\beta>\beta_c$, the system spontaneously magnetizes such that $\mu_\beta^+[\sigma_0]=-\mu_\beta^-[\sigma_0]=m_\beta>0$. Here $\sigma_0$ is the spin at the origin (and hence by translation invariance any spin).

\subsection{Real-Space Renormalization Group Transformations}

Let us now start with a basic introduction on block-spin transformations, one particular class of renormalization group transformations (RGT). We shall consider the Ising model on $\mathds{Z}^2$. In order to apply a block spin transformation we partition the space into $2\times 2$-blocks and assign to the center of each block a vertex $x'$ forming a coarsened version of the same lattice.\smallskip

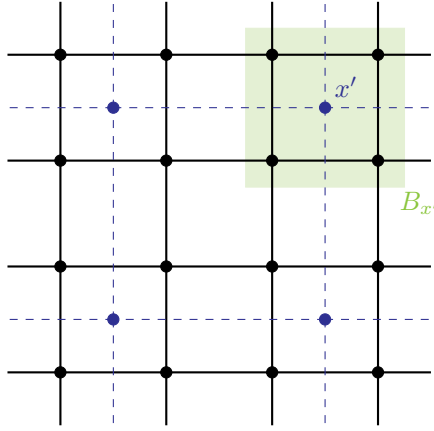
\begin{figure}[h!]
    \centering
    \centering
    \begin{tikzpicture}[scale= 0.7]
        \def \rows {4}  
        \def \cols {4}  
        \def \spacing {2}  
        \def \extend {0.0}  

        \filldraw[LimeGreen!20] ({8.5}, {8.5}) -- ({8.5}, {5.5}) -- ({5.5}, {5.5}) -- ({5.5}, {8.5}) -- cycle;
    
        \foreach \x in {1,...,\cols} {
            \foreach \y in {1,...,\rows} {
                \filldraw (\x*\spacing, \y*\spacing) circle (3 pt);
    
                \ifnum \x<\cols
                    \draw[thick] (\x*\spacing, \y*\spacing) -- ({(\x+1)*\spacing}, \y*\spacing);
                \fi
    
                \ifnum \y<\rows
                    \draw[thick] (\x*\spacing, \y*\spacing) -- (\x*\spacing, {(\y+1)*\spacing});
                \fi
    
                \ifnum \x=\cols
                    \draw[dashed] (\x*\spacing, \y*\spacing) -- ({(\x+\extend)*\spacing}, \y*\spacing);
                \fi
                
                \ifnum \y=\rows
                    \draw[dashed] (\x*\spacing, \y*\spacing) -- (\x*\spacing, {(\y+\extend)*\spacing});
                \fi
    
                \ifnum \x=1
                    \draw[dashed] (\x*\spacing, \y*\spacing) -- ({(\x-\extend)*\spacing}, \y*\spacing);
                \fi
    
                \ifnum \y=1
                    \draw[dashed] (\x*\spacing, \y*\spacing) -- (\x*\spacing, {(\y-\extend)*\spacing});
                \fi
            }
        }

        \filldraw[Blue] ({3}, {3}) circle (3 pt);
        \filldraw[Blue] ({3}, {7}) circle (3 pt);
        \filldraw[Blue] ({7}, {3}) circle (3 pt);
        \filldraw[Blue] ({7}, {7}) circle (3 pt);

        \draw[dashed, Blue] ({3}, {3}) -- ({3}, {7});
        \draw[dashed, Blue] ({3}, {7}) -- ({7}, {7});
        \draw[dashed, Blue] ({7}, {7}) -- ({7}, {3});
        \draw[dashed, Blue] ({7}, {3}) -- ({3}, {3});

        \draw[dashed, Blue] ({3}, {3}) -- ({3}, {1});
        \draw[dashed, Blue] ({3}, {3}) -- ({1}, {3});
        \draw[dashed, Blue] ({7}, {3}) -- ({9}, {3});
        \draw[dashed, Blue] ({7}, {3}) -- ({7}, {1});
        \draw[dashed, Blue] ({3}, {7}) -- ({3}, {9});
        \draw[dashed, Blue] ({3}, {7}) -- ({1}, {7});
        \draw[dashed, Blue] ({7}, {7}) -- ({9}, {7});
        \draw[dashed, Blue] ({7}, {7}) -- ({7}, {9});

        \draw[black][thick] ({2}, {2}) -- ({1}, {2});
        \draw[black][thick] ({2}, {2}) -- ({2}, {1});
        \draw[black][thick] ({2}, {4}) -- ({1}, {4});
        \draw[black][thick] ({2}, {6}) -- ({1}, {6});
        \draw[black][thick] ({2}, {8}) -- ({1}, {8});
        \draw[black][thick] ({2}, {8}) -- ({2}, {9});
        \draw[black][thick] ({4}, {2}) -- ({4}, {1});
        \draw[black][thick] ({6}, {2}) -- ({6}, {1});
        \draw[black][thick] ({8}, {2}) -- ({8}, {1});
        \draw[black][thick] ({8}, {2}) -- ({9}, {2});
        \draw[black][thick] ({8}, {4}) -- ({9}, {4});
        \draw[black][thick] ({8}, {6}) -- ({9}, {6});
        \draw[black][thick] ({8}, {8}) -- ({9}, {8});
        \draw[black][thick] ({8}, {8}) -- ({8}, {9});
        \draw[black][thick] ({6}, {8}) -- ({6}, {9});
        \draw[black][thick] ({4}, {8}) -- ({4}, {9});

        \node[Blue] at (7.4,7.4) {$x'$};
        \node[LimeGreen] at (8.8,5.2) {$B_{x'}$};
    \end{tikzpicture}
    
    \caption{The $2\times 2\to1$ coarse-graining geometry on $\mathds{Z}^2$.}
    \label{fig: 16 to 4 Blocking on the Square Lattice with PBCs}
\end{figure}

The idea is then to define transition rules on how to map the four block-spins to a spin on the vertex $x'$. In order to guarantee some physical relevance of the outcome, we want to put some restriction towards the choice of these rules.
\begin{enumerate}
    \item[R1:] The set of rules should be invariant under symmetry transformations of the block configurations. These include rotations as well as the $\mathds{Z}_2$-action $(s_B,\sigma)\mapsto(-s_B,-\sigma)$.
    \item[R2:] The rules must have a probabilistic interpretation, i.e.\ for any input configuration $s$ we must obtain
    \begin{equation}
        t(\uparrow|s),t(\downarrow|s)\in[0,1]\qquad;\qquad t(\uparrow|s)+t(\downarrow|s)=1.
    \end{equation}
    \item[R3:] $t$ must preserve the monotonicity of the Ising model, i.e $s\leq s'\implies t(\uparrow|s)\leq t(\uparrow|s')$.
    \item[R4:] In case that all four block-spins are parallel, also the center spin must be parallel.
\end{enumerate}
It might not be directly obvious why we want to insist on R4. We will discuss this briefly in Section \ref{Sec3}. Now, these conditions limit the set of rules to a one-parameter family with the following probabilities:

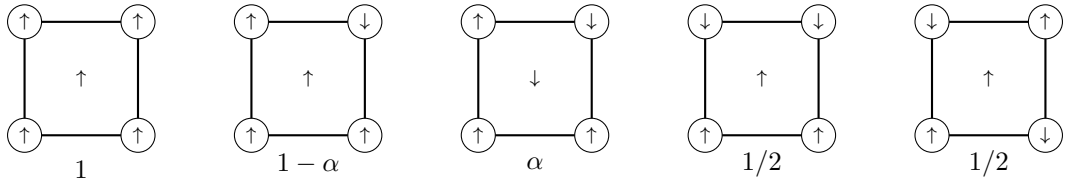
\begin{figure}[htpb!]
    \centering
    \begin{tikzpicture}[scale = 1.5]
        \draw[black] (0,0) circle (0.15);
        \draw[black] (0,1) circle (0.15);
        \draw[black] (1,0) circle (0.15);
        \draw[black] (1,1) circle (0.15);

        \draw[black][thick] (0.15,0) -- (0.85,0);
        \draw[black][thick] (0,0.15) -- (0,0.85);
        \draw[black][thick] (1,0.15) -- (1,0.85);
        \draw[black][thick] (0.15,1) -- (0.85,1);

        \node[] at (0,0) {\scriptsize$\uparrow$};
        \node[] at (0,1) {\scriptsize$\uparrow$};
        \node[] at (1,0) {\scriptsize$\uparrow$};
        \node[] at (1,1) {\scriptsize$\uparrow$};
        \node[] at (0.5,0.5) {\scriptsize$\uparrow$};
        \node[] at (0.5,-0.3) {$1$};

        \begin{scope}[shift={(2,0)}]
            \draw[black] (0,0) circle (0.15);
            \draw[black] (0,1) circle (0.15);
            \draw[black] (1,0) circle (0.15);
            \draw[black] (1,1) circle (0.15);
    
            \draw[black][thick] (0.15,0) -- (0.85,0);
            \draw[black][thick] (0,0.15) -- (0,0.85);
            \draw[black][thick] (1,0.15) -- (1,0.85);
            \draw[black][thick] (0.15,1) -- (0.85,1);
    
            \node[] at (0,0) {\scriptsize$\uparrow$};
            \node[] at (0,1) {\scriptsize$\uparrow$};
            \node[] at (1,0) {\scriptsize$\uparrow$};
            \node[] at (1,1) {\scriptsize$\downarrow$};
            \node[] at (0.5,0.5) {\scriptsize$\uparrow$};
            \node[] at (0.5,-0.25) {$1-\alpha$};
        \end{scope}

        \begin{scope}[shift={(4,0)}]
            \draw[black] (0,0) circle (0.15);
            \draw[black] (0,1) circle (0.15);
            \draw[black] (1,0) circle (0.15);
            \draw[black] (1,1) circle (0.15);
    
            \draw[black][thick] (0.15,0) -- (0.85,0);
            \draw[black][thick] (0,0.15) -- (0,0.85);
            \draw[black][thick] (1,0.15) -- (1,0.85);
            \draw[black][thick] (0.15,1) -- (0.85,1);
    
            \node[] at (0,0) {\scriptsize$\uparrow$};
            \node[] at (0,1) {\scriptsize$\uparrow$};
            \node[] at (1,0) {\scriptsize$\uparrow$};
            \node[] at (1,1) {\scriptsize$\downarrow$};
            \node[] at (0.5,0.5) {\scriptsize$\downarrow$};
            \node[] at (0.5,-0.25) {$\alpha$};
        \end{scope}

        \begin{scope}[shift={(6,0)}]
            \draw[black] (0,0) circle (0.15);
            \draw[black] (0,1) circle (0.15);
            \draw[black] (1,0) circle (0.15);
            \draw[black] (1,1) circle (0.15);
    
            \draw[black][thick] (0.15,0) -- (0.85,0);
            \draw[black][thick] (0,0.15) -- (0,0.85);
            \draw[black][thick] (1,0.15) -- (1,0.85);
            \draw[black][thick] (0.15,1) -- (0.85,1);
    
            \node[] at (0,0) {\scriptsize$\uparrow$};
            \node[] at (0,1) {\scriptsize$\downarrow$};
            \node[] at (1,0) {\scriptsize$\uparrow$};
            \node[] at (1,1) {\scriptsize$\downarrow$};
            \node[] at (0.5,0.5) {\scriptsize$\uparrow$};
            \node[] at (0.5,-0.25) {$1/2$};
        \end{scope}

        \begin{scope}[shift={(8,0)}]
            \draw[black] (0,0) circle (0.15);
            \draw[black] (0,1) circle (0.15);
            \draw[black] (1,0) circle (0.15);
            \draw[black] (1,1) circle (0.15);
    
            \draw[black][thick] (0.15,0) -- (0.85,0);
            \draw[black][thick] (0,0.15) -- (0,0.85);
            \draw[black][thick] (1,0.15) -- (1,0.85);
            \draw[black][thick] (0.15,1) -- (0.85,1);
    
            \node[] at (0,0) {\scriptsize$\uparrow$};
            \node[] at (0,1) {\scriptsize$\downarrow$};
            \node[] at (1,0) {\scriptsize$\downarrow$};
            \node[] at (1,1) {\scriptsize$\uparrow$};
            \node[] at (0.5,0.5) {\scriptsize$\uparrow$};
            \node[] at (0.5,-0.25) {$1/2$};
        \end{scope}

    \end{tikzpicture}
    \caption{The possible transition values for the $2\times 2\to1$ coarse-graining, respecting R1-R4. The values of the other possible block-configurations are determined through R1 and R2.}
    \label{fig:placeholder}
\end{figure}

In order to comply with R3, the allowed range is $\alpha\in[0,1/2]$. Once we have specified the transition probabilities $t_\alpha(\sigma_{x'}|s_{B_{x'}})$ we can define a product probability kernel
\begin{equation}
    \mathcal{T}_\alpha=\prod_{x'\in\mathds{Z}^2}t_\alpha(\sigma_{x'}|s_{B_{x'}}).
\end{equation}
The corresponding RGT is defined as the map
\begin{equation}
    \mu\mapsto \mathcal{T}_\alpha\mu,\qquad \mathcal{T}_\alpha\mu(A):=\int_\Omega \mathcal{T}_\alpha(A|s)\mu(ds).
\end{equation}
This kernel then creates a discrete-time dynamical system on the space of probability measures $\mathscr{M}_1(\Sigma_{\mathds{Z}^2})$. Although these dynamics are always well-defined, the physical theory of RGT goes beyond that and expects that we can always lift these dynamics to a corresponding dynamical system on a space of well-behaved interactions (interactions are the infinite volume equivalent to Hamilton functions). This is not the case, as the authors of \cite{van_Enter_1993} have shown rigorously in many different examples. Generally, this lift is not globally well-defined on any reasonable space of interactions. The issue is that measures compatible with an interaction may lose this property after acting on them with the RGT. In particular, this already happens in some part of the low-temperature region of the nearest-neighbour Ising model itself after a single application of the RGT. It remains an open question, to find the maximal support of any real-space RGT on a suitable space of well-behaved interactions and whether this will include the critical point of the Ising model/general systems at criticality or not. Nevertheless, the measure itself already carries most of the information and even allows to reconstruct the interaction, as long as it is compatible with an interaction of sufficient regularity (again, see \cite{van_Enter_1993}). Thus, it is an interesting question whether we can establish the expected RG-flow, at least as a property of measures. One expects that the scaling limit is given by
\begin{equation}\label{Eq2}
    \lim_{n\to\infty}\mathcal{T}_\alpha^n\mu_\beta^\pm=\begin{cases}
        \mu_{\beta=0}&\beta<\beta_c\\
        \mu_*&\beta=\beta_c\\
        \mu_{\beta=\infty}^\pm&\beta>\beta_c
    \end{cases}
\end{equation}
where $\mu_*$ is some non-trivial fixed point measure (non-trivial may be understood as having non-vanishing truncated correlation functions). It is a straightforward exercise to prove that $\mathcal{T}_\alpha\mu_{\beta=0}=\mu_{\beta=0}$ as well as $\mathcal{T}_\alpha\mu_{\beta=\infty}^\pm=\mu_{\beta=\infty}^\pm$. Subsequently, we will refer to them as the trivial high- resp. low-temperature fixed points. Recall that within $\beta\in[0,\beta_c]$ there is a unique infinite volume Gibbs measure such that the scaling limit in this regime is also independent on the choice of boundary condition. In the low-temperature phase $\beta>\beta_c$ we may use that, by linearity of the RGT, we find
\begin{equation}
    \lim_{n\to\infty}\mathcal{T}_\alpha^n\left(s\mu_\beta^++(1-s)\mu_\beta^-\right)=s\mu_{\beta=\infty}^++(1-s)\mu_{\beta=\infty}^-
\end{equation}
which determines the RG-flow for all possible infinite volume nearest-neighbour Ising states on $\mathds{Z}^2$.

\subsection{Main Results}\label{Sec1.3}

In Section \ref{Sec4} we will establish that for $\alpha\in(0,1/4)$ the RG-flow is in agreement with (\ref{Eq2}) for all $\beta\in\mathds{R}_{\geq0}\backslash\{\beta_c\}$. For $\alpha\in(1/4,1/2)$, we will further prove that
\begin{equation}
    \lim_{n\to\infty}\mathcal{T}_\alpha^n\mu_\beta^\pm=\mu_{\beta=0}
\end{equation}
for all $\beta\in\mathds{R}_{\geq0}$ thus refuting the expected behaviour. Preceding these results, we will also describe the much simpler case $\alpha=1/4$ where the RGT becomes linear and the scaling limit is given by
\begin{equation}
    \lim_{n\to\infty}\mathcal{T}_{1/4}^n\mu_\beta^\pm=\nu_\beta^\pm:=\prod_{x\in\mathds{Z}^2}\left(\frac{1+m_\beta}{2}\delta_{\pm1}(\sigma_x)+\frac{1-m_\beta}{2}\delta_{\mp1}(\sigma_x)\right).
\end{equation}
where $m_\beta$ is the magnetization of $\mu_\beta^+$. At $\beta\leq\beta_c$ we find that $m_\beta=0$ and thus $\nu^\pm_{\beta\leq\beta_c}=\mu_{\beta=0}$. Therefore, also this case does not admit the expected picture in the scaling limit. Crucially, both in the linear case $\alpha=1/4$ as well as in the suplinear case $\alpha\in(1/4,1/2)$ we find that $\lim_{n\to\infty}\mathcal{T}_\alpha^n\mu_{\beta_c}=\mu_{\beta=0}$. One characteristic of the critical point $\mu_{\beta_c}$ is that the decay of correlations with distance follows a power law. This decay behaviour cannot vanish after finitely many iterations of the RGT as we will showcase for $\alpha=1/4$. Therefore, the correlation length (which governs the decay rate for systems with exponentially decaying correlations) remains infinite for the entirety of the flow and only vanishes in the limit.\smallskip

The cases $\alpha=0,1/2$ are expected to show the same behaviour as the $\alpha<1/4$ (resp. $\alpha>1/4)$ cases. However, they are just out of reach for our method. The main open question, which is unfortunately also the most interesting one, is what happens for $\alpha\in[0,1/4)$ and $\beta=\beta_c$. There the simple estimates of correlations fail and some more sophisticated methods will be required to obtain any meaningful results in the future.

\subsection{Geometrical Representation of the Ising Model}

In 1972 Fortuin and Kasteleyn \cite{Fortuin:1971dw} introduced the random cluster model which serves as the baseline for geometrical representations of lattice spin models. The results presented in this section can also be found in \cite{10.1007/978-3-030-32011-9_2} including proofs. The starting point is once more some finite subgraph $G=(V,E)\subset\mathds{Z}^2$. Now, instead of assigning spins to vertices, we consider bond variables living on the edges of $G$. The configuration space therefore is
\begin{equation}
    \Omega_G:=\{0,1\}^E.
\end{equation}
We call an edge $e\in E$ open (resp. closed) if $\omega_e=1$ (resp. $\omega_e=0$)\footnote{In physics literature, one may find the notions of a bond being put ($\omega_e=1$) or not put ($\omega_e=0$).}. For a given bond configuration $\omega\in\Omega_G$ we denote the number of open/closed edges as
\begin{equation}
    o(\omega):=\sum_{e\in E}\ind{\omega_e=1},\qquad c(\omega):=\sum_{e\in E}\ind{\omega_e=0}.
\end{equation}
Further, we also require the count of clusters. A cluster is a maximally, connected domain of open edges. For a precise definition of the number of clusters we must introduce boundary conditions. Similar to the Ising model, we will only be using a very limited variety of boundary conditions. These are 
\begin{enumerate}
    \item[i)] Free boundary conditions, where all edges in $\mathcal{E}\backslash E$ are closed, denoted as $\omega^0$.
    \item[ii)] Wired\footnote{If we also allowed subgraphs with a non-connected complement, the definition of wired boundary conditions would be slightly more subtle as it would collapse all potentially disjoint boundaries into a single cluster. However, for our purposes we have no need for this scenario.} boundary conditions, where all edges in $\mathcal{E}\backslash E$ are open, denoted as $\omega^1$.
\end{enumerate}
This now allows us to define the set of all clusters for a given set of bond variables
\begin{equation}
    \mathcal{K}(\omega^i):=\{C\subseteq\mathcal{V}\,:\,C\,\,\mathrm{is\,\,a\,\,maximally\,\,connected\,\,component\,\,of\,\,open\,\,edges\,\,regarding}\,\,\omega^i\}.
\end{equation}
Finally, we can write down the last ingredient required for the definition of the random cluster model, namely
\begin{equation}
    k(\omega^i):=|\mathcal{K}(\omega^i)|.
\end{equation}
The general random cluster measure on $G$ takes the form
\begin{equation}
    \phi_{G,p,q}^{i}[\mathcal{A}]:=\frac{\sum_{\omega\in\Omega_G}\indnb{\mathcal{A}}(\omega)p^{o(\omega)}(1-p)^{c(\omega)}q^{k(\omega^i)}}{\sum_{\omega\in\Omega_G}p^{o(\omega)}(1-p)^{c(\omega)}q^{k(\omega^i)}}.
\end{equation}
The parameter $q$ can, a priori, be varied continuously in $\mathds{R}_{>0}$. It can be interpreted as a parameter that controls how much the model favours many small clusters for $q>1$ or fewer large clusters for $q<1$. However, there is a special relation between $q\in\mathds{N}_{\geq2}$ and the $q$-state Potts models. Our focus lies entirely on the Ising model which corresponds to the $2$-state Potts model. Therefore, we shall fix $q=2$ and drop the subscript $q$ from now on. This relation can be made manifest in the form of a coupling, sometimes known as the Edwards-Sokal coupling. Let us consider the combined configuration space of spins and bonds
\begin{equation}
    X_G:=\Omega_G\times\Sigma_G,\qquad \Omega_G=\{0,1\}^E,\Sigma_G:=\{-1,1\}^V.
\end{equation}
We then obtain a probability measure
\begin{equation}
    \mathds{P}_{G,p}^\#(\omega,\sigma):=\frac{p^{o(\omega)}(1-p)^{c(\omega)}\cdot\Psi_\omega^\#(\sigma)}{\sum_{\omega,\sigma}p^{o(\omega)}(1-p)^{c(\omega)}\cdot\Psi_\omega^\#(\sigma)}.
\end{equation}
Let us first focus on the free boundary conditions. For any $\omega\in\Omega_G$, $\Psi_\omega^0(\sigma)$ is a measure on $\Sigma_G$ defined as 
\begin{equation}
    \Psi_\omega^0(\sigma):=\indnb{h^0(\omega)}(\sigma),
\end{equation}
where $h^0$ is a map
\begin{equation}
    h^0:\Omega_G\to2^{\Sigma_G},\quad \omega\mapsto\{\sigma\in\Sigma_G\,:\,x\leftrightarrow y\implies\sigma_x=\sigma_y\,\forall x,y\in V\}.
\end{equation}
$x\leftrightarrow y$ is the statement that there exists a path of open edges in $E$ connecting $x$ and $y$ regarding $\omega$. In words, this condition tells us that connected spins must be parallel. Thus, $\Psi_\omega^0(\sigma)$ uniformly distributes spins among the clusters of $\omega^0$. With a quick computation, we see that
\begin{equation}
    \mathds{P}_{G,p}^0(\omega)=\phi_{G,p}^0(\omega)\qquad ;\qquad \mathds{P}_{G,p}^0(\sigma)=\mu_{G,\beta_p}^\mathrm{f}(\sigma),
\end{equation}
for $\beta_p=-\frac{1}{2}\ln(1-p)$. The marginals of the coupling $\mathds{P}_{G,p}^0$ are the random cluster model with free boundary conditions and the nearest-neighbour Ising model with free boundary conditions respectively. For the $\pm$ boundary conditions, we consider
\begin{equation}
    h^\pm:\Omega_G\to2^{\Sigma_G},\quad \omega\mapsto\{\sigma\in\Sigma_G\,:\,x\leftrightarrow y\implies \sigma_x=\sigma_y,x\leftrightarrow\partial G\implies\sigma_x=\pm1\,\forall x,y\in V\}.
\end{equation}
Thus, in addition to connected spins being parallel, we now also demand that everything connected to the boundary has fixed spin $\pm1$. The marginals are then given by
\begin{equation}
    \mathds{P}_{G,p}^\pm(\omega)=\phi_{G,p}^1(\omega)\qquad;\qquad\mathds{P}_{G,p}^\pm(\sigma)=\mu_{G,\beta_p}^\pm(\sigma),
\end{equation}
where the relation between $p$ and $\beta$ remains the same as before. This coupling allows us to transfer computations in the Ising model to questions about connectivity in the random cluster model. Some useful special cases are:
\begin{enumerate}
    \item[i)] The two-point function satisfies\footnote{$\phi_{G,p}^0[x\leftrightarrow y]$ corresponds to the probability of the event $\{x\leftrightarrow y\}$ under the probability measure $\phi_{G,p}^0$. When we write $\{\mathrm{statement}\}$ we refer to the set of all configurations that satisfy said statement.}
    \begin{equation}
        \mu_{G,\beta_p}^\mathrm{f}[\sigma_x\sigma_y]=\mathds{P}_{G,p}^\mathrm{f}[\sigma_x\sigma_y]=\phi_{G,p}^0[x\leftrightarrow y].
    \end{equation}
    This can be proven very quickly via
    \begin{equation}
        \mathds{P}_{G,p}^\mathrm{f}[\sigma_x\sigma_y]=\mathds{P}_{G,p}^\mathrm{f}\big[\mathds{P}_{G,p}^\mathrm{f}[\sigma_x\sigma_y\,|\,\omega]\big].
    \end{equation}
    In case that $x\overset{\omega}{\leftrightarrow}y$ their spins must be aligned and we find $\mathds{P}_{G,p}^\mathrm{f}[\sigma_x\sigma_y\,|\,\omega]=1$. On the other hand, when they are not connected, their spins are independent such that
    \begin{equation}
        \mathds{P}_{G,p}^\mathrm{f}[\sigma_x\sigma_y\,|\,\omega]=\mathds{P}_{G,p}^\mathrm{f}[\sigma_x\,|\,\omega]\cdot\mathds{P}_{G,p}^\mathrm{f}[\sigma_y\,|\,\omega]=0.
    \end{equation}
    This generalizes to arbitrary correlation functions in the following way:
    \begin{equation}
        \mu_{G,\beta_p}^\mathrm{f}[\sigma_A]=\mathds{P}_{G,p}^\mathrm{f}[\sigma_A]=\phi_{G,p}^\mathrm{f}[\mathcal{A}],\qquad \sigma_A:=\prod_{x\in A}\sigma_x,A\subset G
    \end{equation}
    for the event $\mathcal{A}=\{\omega\,:\,|A\cap C|\in2\mathds{N}\,\forall C\in\mathcal{K}(\omega)\}$.
    \item[ii)] The pure $\pm$ boundary conditions of the Ising model relate to wired boundary conditions in the random cluster model. We find
    \begin{equation}
        \mu_{G,\beta_p}^+[\sigma_x]=-\mu_{G,\beta_p}^-[\sigma_x]=\phi_{G,p}^1[x\leftrightarrow\partial G].
    \end{equation}
    
\end{enumerate}
The question now is how the Ising phase transition characterizes itself in the random cluster language. Here, one can prove that for all values of $p\in[0,1]$, the measures $\phi_{G_n,p}^\#$ converge weakly to a measure $\phi_p^\#$ along the limit $G_n\uparrow\mathds{Z}^2$. Further, perhaps a bit surprisingly, we find that $\phi_p^0=\phi_p^1=\phi_p$ for all $p\in[0,1]$. The limit measure $\phi_p$ is translation invariant. The Ising phase transition expresses itself through the following result. Set $p_c:=1-e^{-2\beta_c}$. Then we find that\footnote{The event $\{0\leftrightarrow\infty\}$ is given by $\{0\leftrightarrow\infty\}:=\bigcap_{n\in\mathds{N}}\{\exists x\in\mathds{Z}^2\,:\,\|x\|\geq n\wedge0\leftrightarrow x\}$.}
\begin{enumerate}
    \item[i)] $\phi_p[0\leftrightarrow\infty]=0\,\forall p\in[0,p_c]$. Further, for any $p<p_c$ there exist constants $C_p,c_p>$ such that
    \begin{equation}
        \phi_p[0\leftrightarrow x]\leq C_pe^{-c_p\|x\|}.
    \end{equation}
    This implies that
    \begin{enumerate}
        \item[a)] $\chi_p:=\sum_{x\in\mathcal{V}}\phi_p[0\leftrightarrow x]<\infty$,
        \item[b)] There exist constants $C_p',c_p'>0$ such that $\phi_p[\exists y\,:\,0\leftrightarrow y,\|y\|>r]\leq C_p'e^{-c_p'r}$.
    \end{enumerate}
    At the critical point the decay becomes polynomial and the susceptibility diverges $\chi_{p_c}=\infty$. The scaling is given by $\phi_{p_c}[0\leftrightarrow x]=\|x\|^{-1/4+o(1)}$.
    \item[ii)] In the supercritical (or broken) phase $p>p_c$ there exists $\phi_p$-a.s.\footnote{The abbreviation $\mu$-a.s. stands for $\mu$-almost surely and, in case of probability measures, refers to events that have a probability of one regarding $\mu$.} exactly one infinitely large cluster, i.e.\  $\phi_p[\exists!\,C\,:\,|C|=\infty]=1\,\forall p\in(p_c,1]$ and in turn $\phi_p[0\leftrightarrow\infty]=:\theta_p>0\,\forall p\in(p_c,1]$. Further, this is related to the magnetization via $\theta_p=m_{\beta(p)}$. 
\end{enumerate}
Another property that we will make use of later on is that $\phi_p$ is mixing for all values of $p$ even at $p_c$ (although, of course the mixing rate is much slower at $p_c$). Mixing refers to the property that for local events $\mathcal{A},\mathcal{B}$ we find
\begin{equation}
    \lim_{\|x\|\to\infty}\phi_p[\mathcal{A}\cap\tau_x\mathcal{B}]=\phi_p[\mathcal{A}]\phi_p[\mathcal{B}],
\end{equation}
where $\tau_x$ is a translation by $x$. Before we end this discussion by diving into duality and stochastic domination, we shortly want to establish the following result we will need later on.

\begin{lemma}\label{Lem1}
    Let $p<p_c$. Then we find $\phi_p[|C_0|^n]<\infty$ where $C_0\in\mathcal{K}(\omega)$ denotes the cluster containing the origin.
\end{lemma}

\begin{proof}
    Let us quickly compute
    \begin{equation}
        \phi_p[|C_0|^n]=\sum_{k=1}^\infty k^n\phi_p[|C_0|=k]\leq\sum_{k=1}^\infty k^n\phi_p[|C_0|\geq k].
    \end{equation}
    Additionally, we may use that on the square lattice
    \begin{equation}
        \max_{x\in C_0}\|x\|\geq\frac{\sqrt{|C_0|}}{4}.
    \end{equation}
    Therefore, we find that
    \begin{equation}
        \sum_{k=1}^\infty k^n\phi_p[|C_0|\geq k]\leq\sum_{k=1}^\infty k^n\phi_p\left[\exists y\,:\,0\leftrightarrow y,\|y\|\geq\sqrt{k}/4\right]\leq\sum_{k=1}^\infty k^nC_p'e^{c_p'\sqrt{k}}<\infty.
    \end{equation}
\end{proof}

\subsubsection{Duality}

A very useful concept that is special in two dimensions is the duality transformation for the random cluster model. Let us consider an identical copy of $\mathds{Z}^2$ denoted as $(\mathds{Z}^2)^*$ which we shift by $(1/2,1/2)$. Now, for every edge $e\in\mathcal{E}$ there exists a unique dual edge $e^*\in\mathcal{E}^*$ given by the edge that crosses $e$. Thus, for any subgraph $G\subset\mathds{Z}^2$ we can construct a dual graph $G^*=(V^*,E^*)$ given by all dual edges $e^*\in E^*\iff e\in E$ and their endpoints. Let $\mathds{P}$ denote the coupling
\begin{equation}
    \mathds{P}_G(\omega,\omega^*)=\phi_{G,p}^\#(\omega)\cdot\ind{\omega^*_{e^*}=1-\omega_e}(\omega^*).
\end{equation}
Then the marginal regarding $\omega^*$ satisfies
\begin{equation}
    \mathds{P}_G(\omega^*)=\phi^{1-\#}_{G^*,p^*}(\omega^*)
\end{equation}
for $p^*=(1-p)/(1-p/2)$. In particular, we find that $p=p_c\implies p^*=p_c$ and $p>p_c\iff p^*<p_c$ and vice versa. One useful application of duality is the fact that in the supercritical phase $p>p_c$ the finite clusters themselves behave subcritically. This has the consequence that the susceptibility of finite clusters is itself finite, i.e.\
\begin{equation}
    \chi_p^f:=\sum_{x\in\mathcal{V}}\phi_p[0\leftrightarrow x,0\nlr\infty]<\infty.
\end{equation}
This can be proven using the insight that for any finite cluster $C\in\mathcal{K}(\omega)\backslash\{C_\infty\}$ there must exist a dual cluster $C^*\in\mathcal{K}(\omega^*)$ that fully encircles $C$. Since the dual system is subcritical, the size of dual clusters is exponentially suppressed. Note, however, that the exponential decay of finite clusters in the supercritical phase has been established in all dimensions and is not restricted to the duality argument in $d=2$ (see \cite{Duminil-Copin} for a discussion of this result).

\subsubsection{Stochastic Domination}

The last concept of general Ising/RCM theory we want to look into is stochastic domination. In order to define it properly, we first have to introduce the notion of increasing functions. The single-spin space $\{-1,1\}$ comes with a natural ordering $-1<1$. This then imposes a partial ordering on the configuration space $\Sigma_G$
\begin{equation}
    \sigma\leq\tau\iff\sigma_x\leq\tau_x\,\,\forall x\in V,\quad \sigma,\tau\in\Sigma_G.
\end{equation}
We call a function $f:\Sigma_G\to\mathds{R}$ increasing in case that
\begin{equation}
    \sigma\leq\tau\implies f(\sigma)\leq f(\tau).
\end{equation}
Now let $\mu,\nu\in\mathscr{M}(\Sigma_G)$. We say $\nu$ stochastically dominates $\mu$ denoted as $\mu\preceq\nu$ in case that $\mu[f]\leq\nu[f]$ for all increasing functions $f$. For an arbitrary subgraph $G\subset\mathds{Z}^2$ (here $G$ must not necessarily be connected) and boundary conditions $\tau\leq\tau'$ we find that $\mu_{G,\beta}^\tau\preceq\mu_{G,\beta}^{\tau'}$. Note that increasing $\beta$ is not stochastically ordered, i.e. generally $\mu_\beta^\tau\not\preceq\mu_{\beta'}^{\tau'}$ and $\mu_{\beta'}^{\tau'}\not\preceq\mu_\beta^\tau$ if $\beta\neq\beta'$.

\subsection{Some Additions Regarding Measure Theory}

Although we generally do not want to spend too much time dealing with measure theoretic technicalities, it is not possible to avoid them completely. A much more careful and complete introduction can be found in Chapter 6 of \cite{friedli_velenik_2017}.\\

As long as the configuration space $\Omega$ is finite, things are quite trivial. A probability measure can be identified with a map $\mu:\Omega\to[0,1]$ satisfying $\sum_{\omega\in\Omega}\mu(\omega)=1$. The probability of an event $\mathcal{A}\subseteq\Omega$ is then simply given by
\begin{equation}
    \mu[\mathcal{A}]=\sum_{\omega\in\Omega}\indnb{\mathcal{A}}(\omega)\mu(\omega).
\end{equation}
However, things become more subtle when we consider for example the configuration space $\Sigma_{\mathds{Z}^2}=\{-1,1\}^{\mathds{Z}^2}$. Here, we no longer identify the $\sigma$-algebra with the entire power set $2^{\Sigma_{\mathds{Z}^2}}$ but instead consider the $\sigma$-algebra generated by local events
\begin{equation}
    \mathscr{F}:=\sigma\Big(\{\mathcal{A}\subset\Sigma_{\mathds{Z}^2}\,:\,\exists\Lambda\subset\mathds{Z}^2,|\Lambda|<\infty,\mathcal{A}_\Lambda\subseteq\Sigma_\Lambda\,\,|\,\,s\in\mathcal{A}\iff s|_\Lambda\in\mathcal{A}_\Lambda\}\Big).
\end{equation}
We denote the space of all probability measures on the measurable space $(\Sigma_{\mathds{Z}^2},\mathscr{F})$ as $\mathscr{M}_1(\Sigma_{\mathds{Z}^2})$. Two measures coincide in $\mathscr{M}_1(\Sigma_{\mathds{Z}^2})$ iff they coincide on a generating set, in our case on local events. Equivalently, for $\mu,\nu\in\mathscr{M}_1(\Omega)$ we find
\begin{equation}
    \mu=\nu\iff \mu[f]=\nu[f]\quad \forall f\in\mathcal{B}_\mathrm{loc},
\end{equation}
where $\mathcal{B}_\mathrm{loc}$ is the space of local functions. A function $f:\Sigma_{\mathds{Z}^2}\to\mathds{R}$ is called local in case that there exists some finite subset $\Lambda\subset\mathds{Z}^2$ such that $f(\sigma)=f(\tau)$ whenever $\sigma|_\Lambda=\tau|_\Lambda$. Furthermore, a sequence of measures $(\mu_n)_{n\in\mathds{N}}\subset\mathscr{M}_1(\Sigma_{\mathds{Z}^2})$ converges weakly to some limit measure $\mu\in\mathscr{M}_1(\Sigma_{\mathds{Z}^2})$ in case that
\begin{equation}
    \mu_n[f]\to\mu[f],n\to\infty\quad\forall f\in\mathcal{B}_\mathrm{loc}.
\end{equation}
This relates nicely to correlation functions via the fact that for each $\Lambda\subset\mathds{Z}^2,|\Lambda|<\infty$ the set $\{\sigma_A,A\subseteq\Lambda\}$ is a basis of the space of functions depending only on spins in $\Lambda$. The upshot of this is that
\begin{equation}
    \mu_n\to\mu,n\to\infty\quad\iff\quad \mu_n[\sigma_A]\to\mu[\sigma_A],n\to\infty,\,\,\forall A\subset\mathds{Z}^2,|A|<\infty.
\end{equation}
Therefore, we can prove weak convergence of a sequence of measures purely by proving the convergence of correlation functions. For the remainder of this article, whenever we mention convergence of a sequence of measures we refer to weak convergence. One more useful relation is that we can write any correlation function as a linear combination of increasing functions via
\begin{equation}
    \sigma_A=\sum_{S\subseteq A}(-1)^{|S|}\prod_{x\in S}\ind{\sigma_x=-1}\prod_{y\in A\backslash S}\ind{\sigma_y=+1}=\sum_{S\subseteq A}(-1)^{|S|}\prod_{x\in S}(1-\ind{\sigma_x=+1})\prod_{y\in A\backslash S}\ind{\sigma_y=+1}.
\end{equation}
Therefore, we may instead prove weak convergence via proving convergence on increasing functions.

\section{Real Space RGT as a Random Cluster Model}\label{Sec2}

The general idea of this section is to reinterpret real-space RGTs not as maps acting on measures but as measures set up on a multi-layered graph. Doing so, allows us to compute correlations between coarse spins through considering connectivity arguments similar to the Edwards-Sokal coupling for correlations between Ising spins.

\subsection{RG-Admissible Graphs}\label{Sec2.1}

The first step towards the random cluster representation for block spin transformations is to understand the corresponding graphs we will use for these representations. The key idea is that the graph structure must allow for appropiate cluster connections between the block-spin and the underlying fine spins. The starting point is some countable base graph $L=(V(L),E(L))$. We then construct a graph $G$ with vertex-set
\begin{equation}
    V(G)=V(L)\times\mathds{N}.
\end{equation}
We may denote the vertices on the $k$th level as $V_k(G)=\{(x,k)\,:\,x\in L\}\subset V(G)$. In order to give a proper definition of the edge set of $G$, we first want to decompose it into horizontal and vertical edges
\begin{equation}
    E(G)=E_\mathrm{h}(G)\cup E_\mathrm{v}(G).
\end{equation}
Their definition is given as follows:
\begin{enumerate}
    \item[i)] The horizontal edges $E_\mathrm{h}(G)=E(L)\times\mathds{N}$ satisfy
    \begin{equation}
        (x,k)\sim_G(y,k)\iff x\sim_L y,\quad k\in\mathds{N}.
    \end{equation}
    \item[ii)] $E_\mathrm{v}(G)$ contains vertical edges between vertices $(x,n)$ and $(y,n+1)$, respecting the coarse-graining geometry we want our transformation to implement. Of course, it should be homogeneous throughout the entire graph, i.e.\
    \begin{equation}
        (x,k)\sim (y,k+1)\iff (x,l)\sim (y,l+1)\quad\forall k,l\in\mathds{N}.
    \end{equation}
    Further, we also demand that
    \begin{equation}
        (y,k+1)\sim (x,k)\implies (y',k+1)\not\sim(x,k)\,\,\forall y'\in V(L)\backslash\{y\},
    \end{equation}
    which guarantees that the corresponding blocks of the RGT are non-overlapping.
\end{enumerate}
We shall refer to graphs of this form as RG-admissible graphs. In order to get a better grasp of this definition, let us look at two examples:
\begin{enumerate}
    \item[i)] For the first example we shall consider $L=(\mathds{Z},\{\{z,z+1\}\,:\,z\in\mathds{Z}\})$. Further, the vertical edges are chosen in a way that we obtain a two to one coarse-graining,
    \begin{equation}
        E_{\mathrm{v}}(G)=\{\{(2z,k),(z,k+1)\},\{(2z+1,k),(z,k+1)\}\,:\,z\in\mathds{Z},k\in\mathds{N}\}.
    \end{equation}

    \begin{figure}[h!]
        \centering
        \begin{tikzpicture}
            
            \draw[black][-] (-5.5,0) to (6.5,0);
            \fill[black] (-5,0) circle (0.1);
            \fill[black] (-4,0) circle (0.1);
            \fill[black] (-3,0) circle (0.1);
            \fill[black] (-2,0) circle (0.1);
            \fill[black] (-1,0) circle (0.1);
            \fill[black] (0,0) circle (0.1);
            \fill[black] (1,0) circle (0.1);
            \fill[black] (2,0) circle (0.1);
            \fill[black] (3,0) circle (0.1);
            \fill[black] (4,0) circle (0.1);
            \fill[black] (5,0) circle (0.1);
            \fill[black] (6,0) circle (0.1);
    
            \draw[black][-] (-5.5,1) to (6.5,1);
            \draw[black][-] (-5,0) to (-4.5,1);
            \draw[black][-] (-4,0) to (-4.5,1);
            \draw[black][-] (-3,0) to (-2.5,1);
            \draw[black][-] (-2,0) to (-2.5,1);
            \draw[black][-] (-1,0) to (-0.5,1);
            \draw[black][-] (0,0) to (-0.5,1);
            \draw[black][-] (1,0) to (1.5,1);
            \draw[black][-] (2,0) to (1.5,1);
            \draw[black][-] (3,0) to (3.5,1);
            \draw[black][-] (4,0) to (3.5,1);
            \draw[black][-] (5,0) to (5.5,1);
            \draw[black][-] (6,0) to (5.5,1);
            \fill[black] (-4.5,1) circle (0.1);
            \fill[black] (-2.5,1) circle (0.1);
            \fill[black] (-0.5,1) circle (0.1);
            \fill[black] (1.5,1) circle (0.1);
            \fill[black] (3.5,1) circle (0.1);
            \fill[black] (5.5,1) circle (0.1);
    
            \draw[black][-] (-5.5,2) to (6.5,2);
            \draw[black][-] (-4.5,1) to (-3.5,2);
            \draw[black][-] (-2.5,1) to (-3.5,2);
            \draw[black][-] (-0.5,1) to (0.5,2);
            \draw[black][-] (1.5,1) to (0.5,2);
            \draw[black][-] (3.5,1) to (4.5,2);
            \draw[black][-] (5.5,1) to (4.5,2);
            \fill[black] (-3.5,2) circle (0.1);
            \fill[black] (0.5,2) circle (0.1);
            \fill[black] (4.5,2) circle (0.1);
            
        \end{tikzpicture}
        \caption{A possible RG-admissible graph-structure for a one-dimensional spin-chain.}
        \label{1D}
    \end{figure}
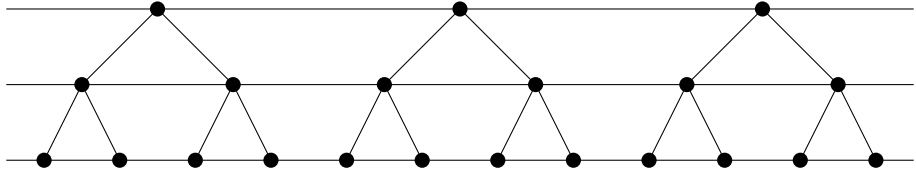

    \item[ii)] The second example takes as base $L=\mathds{Z}^2$. The vertical edges between the different layers represent a $2\times 2\to1$ coarse-graining as illustrated in Figure \ref{SquareRG}. This graph structure allows us to implement the $\mathcal{T}_\alpha$-transformations we studied in the introduction.
\end{enumerate}

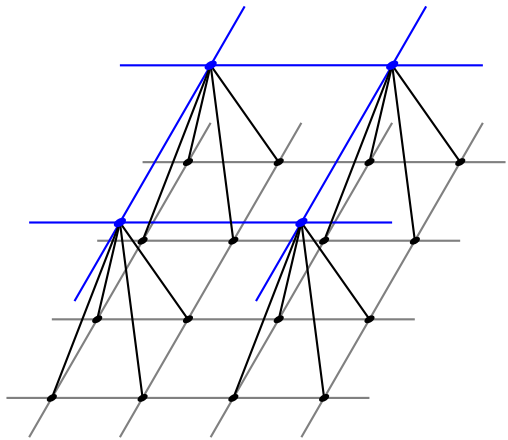
\begin{figure}[h!]
    \centering
    \begin{tikzpicture}[scale=1.2,
        x={(1cm,0cm)}, y={(0.5cm,0.866cm)}, z={(0cm,1cm)} 
        ]
        
        \def\N{3}     
        \def\Nz{1}    
        \def\h{1.5}     
        
        \foreach \i in {0,...,\N} {
          \foreach \j in {0,...,\N} {
            \coordinate (L\i\j) at (\i,\j,0);
          }
        }
        \foreach \i in {0,...,\N} {
          \draw[thick][gray] (L0\i) -- (L\N\i);
          \draw[thick][gray] (L\i0) -- (L\i\N);
        }

        \foreach \i in {0,...,\N}{
            \draw[thick][gray] (3,\i,0) -- (3.5,\i,0);
            \draw[thick][gray] (-0.5,\i,0) -- (0,\i,0);
            \draw[thick][gray] (\i,-0.5,0) -- (\i,0,0);
            \draw[thick][gray] (\i,3,0) -- (\i,3.5,0);
        }

        \foreach \i in {0,...,\Nz} {
          \foreach \j in {0,...,\Nz} {
            \coordinate (U\i\j) at (2*\i+0.5,2*\j+0.5,\h);
          }
        }
        \foreach \i in {0,...,\Nz} {
          \draw[thick][blue] (U0\i) -- (U\Nz\i);
          \draw[thick][blue] (U\i0) -- (U\i\Nz);
        }

        \foreach \i in {0,...,\Nz}{
            \draw[thick][blue] (-0.5,2*\i+0.5,1.5) -- (0.5,2*\i+0.5,1.5);
            \draw[thick][blue] (2.5,2*\i+0.5,1.5) -- (3.5,2*\i+0.5,1.5);
             \draw[thick][blue] (2*\i+0.5,-0.5,1.5) -- (2*\i+0.5,0.5,1.5);
            \draw[thick][blue] (2*\i+0.5,2.5,1.5) -- (2*\i+0.5,3.25,1.5);
        }
        
        \foreach \i in {0,...,1} {
          \foreach \j in {0,...,1} {
            \foreach \a/\b in {0/0,0/1,1/0,1/1} {
              \pgfmathtruncatemacro{\x}{2*\i+\a}
              \pgfmathtruncatemacro{\y}{2*\j+\b}
              \draw[thick] (L\x\y) -- (U\i\j);
            }
          }
        }

        \foreach \i in {0,...,\N} {
            \foreach \j in {0,...,\N} {
                \fill (L\i\j) circle (0.05);
            }
        }
        
        \foreach \i in {0,...,\Nz} {
          \foreach \j in {0,...,\Nz} {
            \fill[blue] (U\i\j) circle (0.06);
          }
        }
        
        \end{tikzpicture}
    \caption{RG-admissible graph-structure for a $2\times 2\to1$ coarse-graining on the $2\text{D}$ square lattice.}
    \label{SquareRG}
\end{figure}

\subsection{Cluster Representation of Real-Space Renormalization Group Transformations}

Once we have determined an RG-admissible graph $G$ we want to work with, we must figure out a way to choose cluster rules representing the coarse-graining we want to implement. As always, we should define the structure first on a finite subgraph of $G$ and then concern ourselves with the infinite volume limit afterwards. For a given $x\in V_n(G),n\in\mathds{N}$ we define the cone $\mathcal{C}_x\subset G$ to be the induced subgraph with vertex set\footnote{$[x_1,...,x_n]$ denotes a path of vertices, i.e.\ vertices such that $x_i\sim x_{i+1}$ .}
\begin{equation}
    V(\mathcal{C}_x)=\{y\in V_k(G)\,:\,k<n,\,\,\exists[(y,k),(z_1,k+1),...,(z_{n-k-1},n-1),(x,n)]\}\cup\{x\}.
\end{equation}
For a given finite subset $\Lambda\subset V_n(G)$ we then define $\mathcal{C}_\Lambda\subset G$ to be the induced subgraph with vertex set
\begin{equation}
    V(\mathcal{C}_\Lambda)=\bigcup_{x\in\Lambda}V(\mathcal{C}_x).
\end{equation}
This leads us to the finite configuration space
\begin{equation}
    X_\Lambda=\Omega_\Lambda\times\Sigma_\Lambda
\end{equation}
where $\Omega_\Lambda=\{0,1\}^{E_\mathrm{v}(\mathcal{C}_\Lambda)}$ and $\Sigma_\Lambda=\{-1,1\}^{V(\mathcal{C}_\Lambda)}$. Since $X_\Lambda$ is finite, in a straightforward manner we can write down a (possibly signed) measure given by
\begin{equation}\label{Eq1}
    \mu_{\phi_\mathrm{v},\Lambda}(\sigma_{\mathcal{C}_\Lambda},\omega_{\mathcal{C}_\Lambda})=\mu_{0,\Lambda}(\sigma^{(0)})\times\Phi_\mathrm{v}(\omega)\times\Psi_\omega(\sigma).
\end{equation}
$\mu_{0,\Lambda}(\sigma^{(0)})$ is the input measure we want to act on with the RGT. $\sigma^{(0)}$ denotes the spin configuration restricted to vertices in $V_0(\mathcal{C}_\Lambda)$. Later on, we will focus on the case $\mu_{0,\Lambda}=\mu_{\mathcal{C}^0_\Lambda,\beta}^\#$ where $\mathcal{C}^0_\Lambda\subset\mathds{Z}^2$ is the subgraph obtained by considering all vertices and edges of $\mathcal{C}_\Lambda$ in the bottom layer. $\Psi_\omega(\sigma)$ is defined similarly to the object we used in the definition of the Edwards-Sokal coupling. It enforces spins connected through $\omega$ to be parallel. The bond part $\Phi_\mathrm{v}$ is generally a signed product measure over all blocks, where we define a block $B_x$ to be a subgraph with vertex set
\begin{equation}
    V(B_x):=\{y\in V_{k-1}(\mathcal{C}_\Lambda)\,:\,y\sim x\}\cup\{x\}\qquad x\in V_k(\mathcal{C}_\Lambda),k>0
\end{equation}
and edge set
\begin{equation}
    E(B_x):=\{\braket{xy}\,:\,y\in B_x\}.
\end{equation}
$\Phi_\mathrm{v}$ then takes the form
\begin{equation}
    \Phi_\mathrm{v}(\omega)=\prod_{\substack{x\in V_k(\mathcal{C}_\Lambda)\\k>0}}\phi_\mathrm{v}(\omega_{B_x}).
\end{equation}
The idea is that the values of coarse spins are fully determined by how they can connect down to the base layer.

\begin{definition}
    Let $\mathcal{T}$ be some RGT applicable to the input measure $\mu_{0,\Lambda}$. We call a measure $\mu_{\phi_\mathrm{v},\Lambda}$ a geometrical representation of the RGT $\mathcal{T}$ in case that for all $k\in\{0,...,n\}$ the marginal measures satisfy
    \begin{equation}
        \mu_{\phi_\mathrm{v},\Lambda}(\sigma_k)=\mathcal{T}^k\mu_{0,\Lambda}(\sigma_k),\qquad \mathcal{T}^0\mu_{0,\Lambda}:=\mu_{0,\Lambda}.
    \end{equation}
\end{definition}

\begin{definition}
    Let $\mu_{\phi_\mathrm{v},\Lambda}$ be a geometrical representation of some RGT $\mathcal{T}$ on some RG-admissible graph $G$. Further, let $x\in V_n(G)$ and $S\subseteq\mathcal{C}_x^0$. We then set
    \begin{equation}
        \Phi_\mathrm{v}(S):=\sum_{\omega\in\Omega_\Lambda}\Phi_\mathrm{v}(\omega)\ind{\{z\in \mathcal{C}_x^0\,:\,x\leftrightarrow z\}=S}(\omega).
    \end{equation}
    The evolution map $f_n$ corresponding to the kernel $\mathcal{T}$ is defined as
    \begin{equation}
        f_n:\mathds{R}^{\mathcal{C}_x^0}\to\mathds{R},\qquad \mathrm{\mathbf{x}}\mapsto\sum_{S\subseteq\mathcal{C}_x^0}\Phi_\mathrm{v}(S)\prod_{i\in S}\mathrm{\mathbf{x}}_i.
    \end{equation}
    In case that $f_n$ is linear, we refer to the map $\mu\mapsto T\mu$ as a linear RGT.
\end{definition}

For a given set $S\subset\mathcal{C}_x^0$ we may define $\textbf{x}_S\in\mathds{R}^{\mathcal{C}_x^0},(\textbf{x}_S)_i:=\ind{i\in S}$. $f_n(\textbf{x}_S)$ then computes the cluster weight of $x$ only being connected to spins in $S$. This is particularly useful in order to compute
\begin{equation}\label{Eq4}
    \mu_{\phi_\mathrm{v},\Lambda}[\sigma_x=\pm1\,|\,\sigma^{(0)}]=f_n\left(\textbf{x}^\pm(\sigma^{(0)})\right),\qquad \textbf{x}^\pm(\sigma^{(0)}):=\textbf{x}_{\{i\in\mathcal{C}_x^0\,:\,\sigma_i^{(0)}=\pm1\}}.
\end{equation}
Recall that in order to prove weak convergence, we want to compute correlation functions. Therefore, let us consider some $A\subset\Lambda,\Lambda\subset V_n(G)$. We then find
\begin{equation}
    \mathcal{T}^n\mu_{0,\Lambda}[\sigma_A]=\mu_{\phi_\mathrm{v},\Lambda}[\sigma_A]=\mu_{\phi_\mathrm{v},\Lambda}[\mu_{\phi_\mathrm{v},\Lambda}[\sigma_A\,|\,\sigma^{(0)}]]=\mu_{0,\Lambda}[\mu_{\phi_\mathrm{v},\Lambda}[\sigma_A|\sigma^{(0)}]],
\end{equation}
where we have used the law of total expectation for the second equality. Crucially, once we condition on the spin configuration at the bottom, $\mu_{\phi_\mathrm{v},\Lambda}[\cdot\,|\,\sigma^{(0)}]$ becomes a pure product measure. In particular, we find
\begin{equation}
    \mu_{\Phi_\mathrm{v},\Lambda}[\sigma_A|\sigma^0]=\sum_{S\subseteq A}(-1)^{|S|}\prod_{i\in S}\mu_{\Phi_\mathrm{v},\Lambda}[\sigma_i=-1|\sigma^0]\prod_{j\in S^c}\mu_{\Phi_\mathrm{v},\Lambda}[\sigma_j=+1|\sigma^0].
\end{equation}
Here we can insert the identity from Equation (\ref{Eq4}) giving us (omitting the $\sigma^{(0)}$-dependence from the notation)
\begin{equation}
    \mathcal{T}^n\mu_{0,\Lambda}[\sigma_A]=\sum_{S\subseteq A}(-1)^{|A|}\mu_{0,\Lambda}\left[\prod_{i\in S}f_n(\textbf{x}^-_i)\prod_{j\in S^c}f_n(\textbf{x}^+_j)\right].
\end{equation}
The $f_n$ are simply local functions. In case that the $\mu_{0,\Lambda}$ converge weakly to some measure $\mu_0$ for $\Lambda\uparrow\mathds{Z}^2$, taking the infinite volume limit of the renormalized quantity becomes trivial 
\begin{equation}
    \lim_{\Lambda\uparrow\mathds{Z}^2}\mathcal{T}^n\mu_{0,\Lambda}[\sigma_A]=\sum_{S\subseteq A}(-1)^{|A|}\mu_{0}\left[\prod_{i\in S}f_n(\textbf{x}^-_i)\prod_{j\in S^c}f_n(\textbf{x}^+_j)\right].
\end{equation}
Therefore, instead of computing correlation functions of the renormalized measures, we can compute
expectation values of local observables of the input measure $\mu_0$. The $f_n$ themselves can be thought of as a signal propagation process. $\textbf{x}^+(\sigma^{(0)})$ assigns to each vertex in $\mathcal{C}_x^0$ a gate that is active whenever the corresponding spin is $+$. The cluster weights $\phi_\mathrm{v}$ control the transmission of this signal up the layers. The probability that the spin $\sigma_x$ at the top of $\mathcal{C}_x^0$ is also $+$ is then given by the strength of the processed signal $f_n(\textbf{x}^+)$.

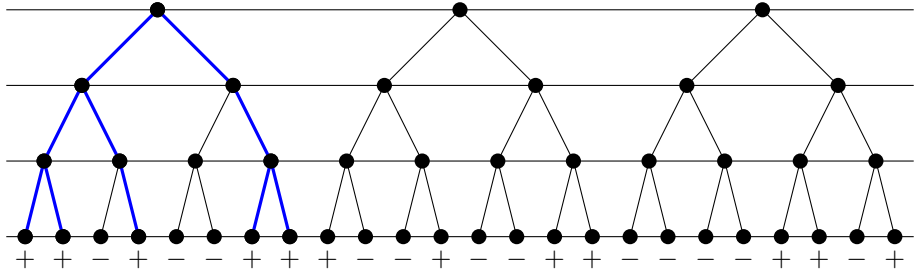
\begin{figure}[h!]
    \centering
    \begin{tikzpicture}

        \draw[black][-] (-5.5,-1) to (6.5,-1);
        \fill[black] (-5.25,-1) circle (0.1);
        \fill[black] (-4.75,-1) circle (0.1);
        \fill[black] (-4.25,-1) circle (0.1);
        \fill[black] (-3.75,-1) circle (0.1);
        \fill[black] (-3.25,-1) circle (0.1);
        \fill[black] (-2.75,-1) circle (0.1);
        \fill[black] (-2.25,-1) circle (0.1);
        \fill[black] (-1.75,-1) circle (0.1);
        \fill[black] (-1.25,-1) circle (0.1);
        \fill[black] (-0.75,-1) circle (0.1);
        \fill[black] (-0.25,-1) circle (0.1);
        \fill[black] (0.25,-1) circle (0.1);
        \fill[black] (0.75,-1) circle (0.1);
        \fill[black] (1.25,-1) circle (0.1);
        \fill[black] (1.75,-1) circle (0.1);
        \fill[black] (2.25,-1) circle (0.1);
        \fill[black] (2.75,-1) circle (0.1);
        \fill[black] (3.25,-1) circle (0.1);
        \fill[black] (3.75,-1) circle (0.1);
        \fill[black] (4.25,-1) circle (0.1);
        \fill[black] (4.75,-1) circle (0.1);
        \fill[black] (5.25,-1) circle (0.1);
        \fill[black] (5.75,-1) circle (0.1);
        \fill[black] (6.25,-1) circle (0.1);

        \foreach \x in {-5,...,6} {
            \draw[black, -] ({\x - 0.25}, -1) -- (\x, 0);
            \draw[black, -] ({\x + 0.25}, -1) -- (\x, 0);
        }

        \draw[black][-] (-5.5,0) to (6.5,0);
        \fill[black] (-5,0) circle (0.1);
        \fill[black] (-4,0) circle (0.1);
        \fill[black] (-3,0) circle (0.1);
        \fill[black] (-2,0) circle (0.1);
        \fill[black] (-1,0) circle (0.1);
        \fill[black] (0,0) circle (0.1);
        \fill[black] (1,0) circle (0.1);
        \fill[black] (2,0) circle (0.1);
        \fill[black] (3,0) circle (0.1);
        \fill[black] (4,0) circle (0.1);
        \fill[black] (5,0) circle (0.1);
        \fill[black] (6,0) circle (0.1);

        \draw[black][-] (-5.5,1) to (6.5,1);
        \draw[black][-] (-5,0) to (-4.5,1);
        \draw[black][-] (-4,0) to (-4.5,1);
        \draw[black][-] (-3,0) to (-2.5,1);
        \draw[black][-] (-2,0) to (-2.5,1);
        \draw[black][-] (-1,0) to (-0.5,1);
        \draw[black][-] (0,0) to (-0.5,1);
        \draw[black][-] (1,0) to (1.5,1);
        \draw[black][-] (2,0) to (1.5,1);
        \draw[black][-] (3,0) to (3.5,1);
        \draw[black][-] (4,0) to (3.5,1);
        \draw[black][-] (5,0) to (5.5,1);
        \draw[black][-] (6,0) to (5.5,1);
        \fill[black] (-4.5,1) circle (0.1);
        \fill[black] (-2.5,1) circle (0.1);
        \fill[black] (-0.5,1) circle (0.1);
        \fill[black] (1.5,1) circle (0.1);
        \fill[black] (3.5,1) circle (0.1);
        \fill[black] (5.5,1) circle (0.1);

        \draw[black][-] (-5.5,2) to (6.5,2);
        \draw[black][-] (-4.5,1) to (-3.5,2);
        \draw[black][-] (-2.5,1) to (-3.5,2);
        \draw[black][-] (-0.5,1) to (0.5,2);
        \draw[black][-] (1.5,1) to (0.5,2);
        \draw[black][-] (3.5,1) to (4.5,2);
        \draw[black][-] (5.5,1) to (4.5,2);
        \fill[black] (-3.5,2) circle (0.1);
        \fill[black] (0.5,2) circle (0.1);
        \fill[black] (4.5,2) circle (0.1);

        \foreach \y in {-5.25,-4.75,-3.75,-2.25,-1.75,-1.25,0.25,1.75,2.25,4.75,5.25,6.25}{
            \node[] at (\y,-1.3) {$+$};
        }

        \foreach \z in {-4.25,-3.25,-2.75,-0.75,-0.25,0.75,1.25,2.75,3.25,3.75,4.25,5.75}{
            \node[] at (\z,-1.3) {$-$};
        }

        \draw[blue, very thick, -] (-5.25,-1) to (-5,0);
        \draw[blue, very thick, -] (-4.75,-1) to (-5,0);
        \draw[blue, very thick, -] (-3.75,-1) to (-4,0);
        \draw[blue, very thick, -] (-2.25,-1) to (-2,0);
        \draw[blue, very thick, -] (-1.75,-1) to (-2,0);
        \draw[blue, very thick, -] (-5,0) to (-4.5,1);
        \draw[blue, very thick, -] (-4,0) to (-4.5,1);
        \draw[blue, very thick, -] (-2,0) to (-2.5,1);
        \draw[blue, very thick, -] (-4.5,1) to (-3.5,2);
        \draw[blue, very thick, -] (-2.5,1) to (-3.5,2);

        \fill[black] (-5.25,-1) circle (0.1);
        \fill[black] (-4.75,-1) circle (0.1);
        \fill[black] (-3.75,-1) circle (0.1);
        \fill[black] (-2.25,-1) circle (0.1);
        \fill[black] (-1.75,-1) circle (0.1);
        \fill[black] (-5,0) circle (0.1);
        \fill[black] (-4,0) circle (0.1);
        \fill[black] (-2,0) circle (0.1);
        \fill[black] (-4.5,1) circle (0.1);
        \fill[black] (-2.5,1) circle (0.1);
        \fill[black] (-3.5,2) circle (0.1);

    \end{tikzpicture}
    \caption{$1\mathrm{D}$ visualization of the signal from the base spins travelling up to the top.}
    \label{1D with Clusters}
\end{figure}

In case that $\mathcal{T}$ is an ordinary probability kernel, we deduce from Equation (\ref{Eq4}) that $f_n\left(\textbf{x}^+(\sigma^{(0)})\right)\in[0,1]$ as well as $f_n\left(\textbf{x}^+(\sigma^{(0)})\right)+f_n\left(\textbf{x}^-(\sigma^{(0)})\right)=1$. Combining these relations with the fact that the $f_n$ are multilinear polynomials, we deduce that either $f_n$ is linear or $f_n$ contains some negative coefficients. The second case subsequently implies that some of the cluster weights themselves have to be negative such that $\phi_\mathrm{v}$ is a signed measure. In the next two sections we will conclude that linear RGTs do not induce the expected RG-flow such that we have no other choice than to work with signed measures for the bond part $\phi_\mathrm{v}$. Unfortunately, dealing with signed measures makes it more difficult to use pure connectivity arguments (which usually rely on estimates of connection events) as we must always keep analytic control over the cancellations going on.

\section{Decimation}\label{Sec3}

Although decimation is outside of the family of RGTs we introduced in the introduction, it does serve as a very simple example to showcase the general idea. The base graph remains $L=\mathds{Z}^2$. However, we alter the choice of vertical edges to represent the decimation idea. One suitable choice is given by
\begin{equation}
    E_\mathrm{v}\left(G\right)=\left\{\{(z,n),(2z,n-1)\}\,:\,z\in\mathds{Z}^2,n\in\mathds{N}\right\}.
\end{equation}
Let us for now dismiss R4 from our set of transition rules and only impose conditions R1-R3. In that case, there are two transition probabilities (we denote the transition probabilities for a single block as $[s\to\sigma]:=t(\sigma|s)$)
\begin{align}
    [+\to+]&=1-\kappa\\
    [+\to-]&=\kappa
\end{align}
for some $\kappa\in[0,1/2]$. Now, let's match those with the possible bond configurations. Of course, each block only consists of a single vertical edge and thus the single-block measure takes the form
\begin{equation}
    \phi_\kappa(\omega)=A\ind{\omega=1}+B\ind{\omega=0}.
\end{equation}
In the $[+\to+]$ scenario we can either set the edge to open or closed while in the $[+\to-]$ scenario the edge must be closed. We find the solution
\begin{equation}
    \begin{aligned} 
        &[+\to+]=1-\kappa=A+B \\ 
        &[+\to-]=\kappa=B 
    \end{aligned} \implies \phi_\kappa(\omega)=(1-2\kappa)\delta_{\omega1}+\kappa\delta_{\omega0}.
\end{equation}
The corresponding evolution map is given by
\begin{equation}
    f_n(x)=(1-2\kappa)^nx+\frac{1}{2}\left(1-(1-2\kappa)^n\right).
\end{equation}
For $\kappa>0$ we find $f_n\to\frac{1}{2},n\to\infty$. Consequently, we must find
\begin{equation}
    \mathcal{T}^n_\kappa\mu_{0}=\mu_{\beta=0}
\end{equation}
for any choice of $\mu_0$. By introducing a non-zero probability that a coarse spin does not connect to anything we do introduce a non-zero probability that the coarse spin does not contain any information on the system below. As we do not want the coarse-spins to completely forget the physics of the scales below, we do want to hold on to R4. Imposing it here forces $\kappa=0$ which leads to the evolution map $f_n(x)=x$. The correlation functions follow from
\begin{equation}
    \mathcal{T}^n\mu_{\beta}^+[\sigma_S]=\mu_\beta^+[\sigma_{S_0^n}],\qquad S_0^n:=\{y\in\mathcal{C}_x^0\,:\,x\in S\},
\end{equation}
where we note that each of the $\mathcal{C}_x^0$ consists of just a single vertex. The distance between any pair of vertices in $S_0^n$ is at least $2^n$. Therefore, the probability that two of them are in the same finite cluster goes to zero in the $n\to\infty$ limit, even at $\beta=\beta_c$. We conclude that
\begin{equation}
    \lim_{n\to\infty}\mathds{P}_\beta^+[\sigma_{S_0^n}]=\lim_{n\to\infty}\mathds{P}_\beta^+[S_0^n\subset C_\infty]=\lim_{n\to\infty}\phi_{p(\beta)}^1[S_0^n\subset C_\infty].
\end{equation}
Since $\phi_{p(\beta)}^1$ is mixing for all values of $p\in[0,1]$ (see \cite{10.1007/978-3-030-32011-9_2}) local events at infinite spatial separation become independent such that
\begin{equation}
    \lim_{n\to\infty}\phi_{p(\beta)}^1[S_0^n\subset C_\infty]=\lim_{n\to\infty}\prod_{x\in S_0^n}\phi_{p(\beta)}^1[x\leftrightarrow\infty]=\theta_\beta^{|S|}.
\end{equation}
These are exactly the correlation functions of the measures $\nu_\beta^+$ from the introduction. This derivation could also be carried out for $\nu_\beta^-$ leading to the correlations $\lim_{n\to\infty}\mathcal{T}_{\kappa=0}^n\mu_\beta^-[\sigma_S]=(-\theta_\beta)^{|S|}$. All in all, we conclude that
\begin{equation}
    \mathcal{T}_{\kappa=0}^n\mu_\beta^\pm\to\nu_\beta^\pm,n\to\infty\qquad \forall\beta\in\mathds{R}_{\geq0}.
\end{equation}
In particular, the critical point $\mu_{\beta_c}$ converges to $\nu_{\beta_c}=\mu_{\beta=0}$ in the scaling limit. Thus, decimation does not reproduce the expected RG-flow nor does it exhibit a non-trivial limit for $\beta_c$. However, in this case it is particularly easy to prove that the correlation length remains infinite after finitely many RG-steps when we start at $\beta_c$. Let us compute the two-point function
\begin{equation}
    \mathcal{T}^n\mu_{\beta_c}[\sigma_x\sigma_y]=\phi_{p_c}[x_0^n\leftrightarrow y_0^n]=(2^n\|x-y\|)^{-1/4+o(1)}.
\end{equation}
Thus, the decay remains polynomial with the same exponent for all $n\in\mathds{N}$. The factor that drives the vanishing of correlation in the limit is the constant $2^{-1/4n+o(n)}$.

\section{$2\times 2\to1$-RGT on the Square Lattice}\label{Sec4}

After working through the very simple example of the decimation map, we want to turn our attention to the family of RGTs we considered in the introduction. In particular, we are interested whether these RGTs are more in line with the expected picture of RG-flow we painted in the introduction. The RG-admissible graph of choice is the one from Example ii) in Section \ref{Sec2.1}. Recall that single blocks contain four vertical edges.

\subsection{$\alpha=1/4$}

Before we consider the general case, we want to investigate the special case obtained by setting $\alpha=1/4$. The single-block measure is given by
\begin{equation}
    \phi_\mathrm{v}(\omega)=\begin{cases}
        1/4&|\{e\,:\omega_e=1\}|=1\\
        0&\mathrm{otherwise}
    \end{cases}.
\end{equation}
which looks simply like a space averaged version of decimation. This is again a case of a linear RGT with the evolution map being given by
\begin{equation}
    f_n(\textbf{x})=\frac{1}{|\mathcal{C}_x^0|}\sum_{i\in\mathcal{C}_x^0}\textbf{x}_i.
\end{equation}
Let us compute the expectation of $f_n$
\begin{equation}
    \mu_\beta^+[f_n]=\frac{1}{|\mathcal{C}_x^0|}\sum_{x\in\mathcal{C}_x^0}\frac{1+\phi_{p(\beta)}[x\leftrightarrow\infty]}{2}=\frac{1+\theta_\beta}{2}.
\end{equation}
The variance is further given by $\mathrm{Var}_\beta^+(f_n)=\frac{1}{4}\mathrm{Var}_\beta^+(m_{\mathcal{C}_x^0})$ where $m_{\mathcal{C}_x^0}:=\frac{1}{|\mathcal{C}_x^0|}\sum_{z\in\mathcal{C}_x^0}\sigma_z$ is the magnetization density in the volume $\mathcal{C}_x^0$. In the limit $\mathcal{C}_x^0\uparrow\mathds{Z}^2$ the variance of $m_{\mathcal{C}_x^0}$ regarding the pure state $\mu_\beta^+$ vanishes for all values $\beta\in\mathds{R}$ and hence $\lim_{n\to\infty}\mathrm{Var}(f_n)=0$. This implies that in the limit $n\to\infty$ $f_n$ takes on the value $(1+\theta_\beta)/2$ $\mu_\beta^+$-a.s. We conclude that\footnote{Note that $\nu^\pm_{\beta=\infty}=\mu^\pm_{\beta=\infty}$.}
\begin{equation}
    \lim_{n\to\infty}\mathcal{T}^n_{\alpha=1/4}\mu_\beta^\pm=\nu_\beta^\pm,\qquad \beta\in\mathds{R}_{\geq0}\cup\{\infty\}.
\end{equation}
Thus, rather unsurprisingly the scaling limit of the $\alpha=1/4$ transformation coincides with the scaling limit of decimation. Although not a formal proof, these two examples should be fairly convincing that any linear RGT that preserves the Ising symmetries and monotonicity will always produce $\nu_\beta^\pm$ in the scaling limit, even at $\beta=\beta_c$. Therefore, in order to obtain an RG-flow in line with the expected saddle point behaviour, it is inevitable to turn our attention to non-linear RGTs. Let us also quickly highlight that
\begin{equation}
    \mathcal{T}_{\alpha=1/4}^n\mu_{\beta_c}[\sigma_x\sigma_y]=\frac{1}{|\mathcal{C}_x^0|^2}\sum_{z\in\mathcal{C}_x^0,w\in\mathcal{C}_y^0}\phi_{p_c}[z\leftrightarrow w]=\frac{1}{|\mathcal{C}_x^0|^2}\sum_{z\in\mathcal{C}_x^0,w\in\mathcal{C}_y^0}\|z-w\|^{-1/4+o(1)}.
\end{equation}
For an arbitrary pairing $z\in\mathcal{C}_x^0,w\in\mathcal{C}_y^0$ we find that $2^n(\|x-y\|-1)\leq\|z-w\|\leq2^n(\|x-y\|+\sqrt{2})$. Once $\|x-y\|$ becomes large we can neglect the additional displacement of the vertices in $\mathcal{C}_x^0,\mathcal{C}_y^0$ relative to the respective centres. Thus, also the scaling of the two-point function coincides with the result we obtained for the decimation map, i.e. $\mathcal{T}_{\alpha=1/4}^n\mu_{\beta_c}[\sigma_x\sigma_y]=(2^n\|x-y\|)^{-1/4+o(1)}$. Let us iterate again that this remains polynomial for all finite $n\in\mathds{N}$ and only vanishes in the limit.

\subsection{General $\alpha$}

Now that we have worked out the $\alpha=1/4$ case we shall start considering the general case. In order to compute the cluster weights, we must first figure out all possible bond configurations:

\begin{figure}[htpb!]
    \centering
    \begin{tikzpicture}[scale = 0.8]
        
        \draw[black][thick][-] (-1,0) to (0,1);
        \draw[black][-] (-1,0) to (1,0);
        \draw[black][-] (-1,0) to (-1,2);
        \draw[black][-] (1,0) to (1,2);
        \draw[black][-] (-1,2) to (1,2);
        \fill[black] (-1,0) circle (0.1);
        \fill[black] (1,0) circle (0.1);
        \fill[black] (-1,2) circle (0.1);
        \fill[black] (1,2) circle (0.1);
        \fill[blue] (0,1) circle (0.1);
        \node[] at (0,-0.35) {$A$};

        \begin{scope}[shift={(3.5,0)}]
            \draw[black][thick][-] (-1,0) to (0,1);
            \draw[black][thick][-] (1,0) to (0,1);
            \draw[black][-] (-1,0) to (1,0);
            \draw[black][-] (-1,0) to (-1,2);
            \draw[black][-] (1,0) to (1,2);
            \draw[black][-] (-1,2) to (1,2);
            \fill[black] (-1,0) circle (0.1);
            \fill[black] (1,0) circle (0.1);
            \fill[black] (-1,2) circle (0.1);
            \fill[black] (1,2) circle (0.1);
            \fill[blue] (0,1) circle (0.1);
            \node[] at (0,-0.35) {$B$};
        \end{scope}

        \begin{scope}[shift={(7,0)}]
            \draw[black][thick][-] (-1,0) to (0,1);
            \draw[black][thick][-] (1,2) to (0,1);
            \draw[black][-] (-1,0) to (1,0);
            \draw[black][-] (-1,0) to (-1,2);
            \draw[black][-] (1,0) to (1,2);
            \draw[black][-] (-1,2) to (1,2);
            \fill[black] (-1,0) circle (0.1);
            \fill[black] (1,0) circle (0.1);
            \fill[black] (-1,2) circle (0.1);
            \fill[black] (1,2) circle (0.1);
            \fill[blue] (0,1) circle (0.1);
            \node[] at (0,-0.35) {$C$};
        \end{scope}

        \begin{scope}[shift={(10.5,0)}]
            \draw[black][thick][-] (-1,0) to (0,1);
            \draw[black][thick][-] (1,0) to (0,1);
            \draw[black][thick][-] (-1,2) to (0,1);
            \draw[black][-] (-1,0) to (1,0);
            \draw[black][-] (-1,0) to (-1,2);
            \draw[black][-] (1,0) to (1,2);
            \draw[black][-] (-1,2) to (1,2);
            \fill[black] (-1,0) circle (0.1);
            \fill[black] (1,0) circle (0.1);
            \fill[black] (-1,2) circle (0.1);
            \fill[black] (1,2) circle (0.1);
            \fill[blue] (0,1) circle (0.1);
            \node[] at (0,-0.35) {$D$};
        \end{scope}

        \begin{scope}[shift={(14,0)}]
            \draw[black][thick][-] (-1,0) to (0,1);
            \draw[black][thick][-] (1,0) to (0,1);
            \draw[black][thick][-] (-1,2) to (0,1);
            \draw[black][thick][-] (1,2) to (0,1);
            \draw[black][-] (-1,0) to (1,0);
            \draw[black][-] (-1,0) to (-1,2);
            \draw[black][-] (1,0) to (1,2);
            \draw[black][-] (-1,2) to (1,2);
            \fill[black] (-1,0) circle (0.1);
            \fill[black] (1,0) circle (0.1);
            \fill[black] (-1,2) circle (0.1);
            \fill[black] (1,2) circle (0.1);
            \fill[blue] (0,1) circle (0.1);
            \node[] at (0,-0.35) {$E$};
        \end{scope}
        
    \end{tikzpicture}
    \caption{The possible bond configurations on the vertical edges of a single block.}
    \label{fig:placeholder}
\end{figure}
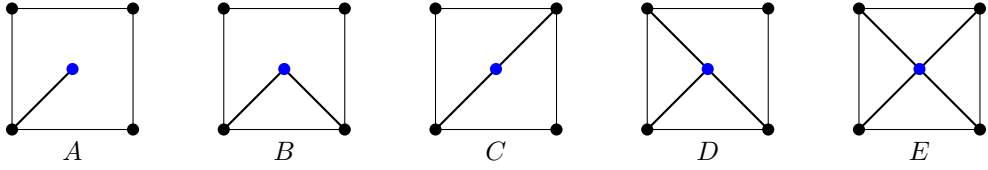

The next step is to match these cluster weights with the transition probabilities of the $\mathcal{T}_\alpha$-kernels
\begin{align}
    \left[\begin{matrix}
        +&+\\+&+
    \end{matrix}\,\,\to\,\,+\right]&=1=4A+4B+2C+4D+E,\nonumber\\
    \left[\begin{matrix}
        +&+\\+&-
    \end{matrix}\,\,\to\,\,+\right]&=1-\alpha=3A+2B+C+D,\nonumber\\
    \left[\begin{matrix}
        +&+\\+&-
    \end{matrix}\,\,\to\,\,-\right]&=\alpha=A,\\
    \left[\begin{matrix}
        +&+\\-&-
    \end{matrix}\,\,\to\,\,+\right]&=1/2=2A+B,\nonumber\\
    \left[\begin{matrix}
        +&-\\-&+
    \end{matrix}\,\,\to\,\,+\right]&=1/2=2A+C.\nonumber
\end{align}
Solving this simple system of equations leads us to
\begin{align}
    A&=\alpha,\nonumber\\
    B&=1/2-2\alpha,\nonumber\\
    C&=1/2-2\alpha,\\
    D&=1-\alpha-3\alpha-3/2+6\alpha=-1/2+2\alpha,\nonumber\\
    E&=1-4\alpha-1+4\alpha=0,\nonumber
\end{align}
where we find that $B=C$ such that we will from now on denote all edge-configurations with two open edges as $B$. The corresponding evolution map for a single step is given by\footnote{Unless we require an explicit choice of $\alpha$, we will not adopt it into the notation.}
\begin{equation}
    f_1(x_1,...,x_4)=\alpha\sum_ix_i+\left(\frac{1}{2}-2\alpha\right)\left(\sum_{i<j}x_ix_j-\sum_{i<j<k}x_ix_jx_k\right).
\end{equation}
For multiple steps, we can compute the evolution map through the following simple recursion: $f_n(\textbf{x})=f(f_{n-1}(\textbf{x}_1),...,f_{n-1}(\textbf{x}_4))$ where we denote $f:=f_1$ from this point onwards. In order to compute correlation functions using the evolution map, it is desirable to establish some basic properties of $f_n$. Recall that the relation $\mu_{\alpha,n,\Lambda,\beta}[\sigma_x=+1|\sigma^{(0)}]=f_n\left(\textbf{x}^+(\sigma^{(0)})\right)$ implies
\begin{enumerate}
    \item[i)] $f_n\left(\textbf{x}^+(\sigma^{(0)})\right)\in[0,1]$,
    \item[ii)] $f_n\left(\textbf{x}^+(\sigma^{(0)})\right)+f_n\left(\textbf{x}^+(-\sigma^{(0)})\right)=f_n\left(\textbf{x}^+(\sigma^{(0)})\right)+f_n\left(\textbf{x}^-(\sigma^{(0)})\right)=1$.
\end{enumerate}
A third property we profit from as well is the following.
\begin{lemma}
    The map $f_n(\mathrm{\mathbf{x}}^+(\sigma))$ is increasing in $\sigma$.
\end{lemma}

\begin{proof}
    It is obvious that $\textbf{x}^+(\sigma)$ itself is an increasing function in $\sigma$. Thus, it is sufficient to prove that the $f_n$ are increasing on $[0,1]^{\mathcal{C}_x^0}$. Therefore, we shall compute
    \begin{equation}
        \frac{\partial f}{\partial x_i}=\alpha+\left(\frac{1}{2}-2\alpha\right)\left(\sum_{j\neq i}x_j-\sum_{j<k\neq i}x_ix_k\right)\geq0\quad\forall x\in[0,1]
    \end{equation}
    which implies that $f$ itself is an increasing function. Of course, by induction we immediately find that $f_{n+1}(\textbf{x}^+)=f(f_n(\textbf{x}^+_1),...,f_n(\textbf{x}^+_4))$ is increasing as well.
\end{proof}

At last we also want to introduce
\begin{equation}
    \delta f_n:=\sup_{i\in\mathcal{C}_x^0}\sup_{\textbf{x}\in\{0,1\}^{\mathcal{C}_x^0}}|f_n(\textbf{x}_{i=+1})-f_n(\textbf{x}_{i=-1})|,
\end{equation}
which describes the maximal increase $f_n$ can experience when flipping a single spin. Since $f_n$ is multilinear, we find that
\begin{equation}
    f_n(\textbf{x}_{i=+1})-f_n(\textbf{x}_{i=-1})=\frac{\partial f_n}{\partial x_i}(\textbf{x}).
\end{equation}
Therefore, we may compute the $\delta f_n$ via
\begin{equation}\label{Eq3}
    \frac{\partial f}{\partial x_i}(\textbf{x})\leq\max\left\{\alpha,\left(\frac{1}{2}-\alpha\right)\right\}\implies\delta f_n=\max\left\{\alpha^n,\left(\frac{1}{2}-\alpha\right)^n\right\}.
\end{equation}\smallskip

The path towards establishing the results presented in Section \ref{Sec1.3} is as follows. First, we prove that for $\beta\neq\beta_c$ the variance $\mu_\beta^\pm\big[(f_n(\textbf{x}^+)-\mu_\beta^\pm[f_n(\textbf{x}^+)])^2\big]$ vanishes in the limit $n\to\infty$. This implies that in the limit the value of $f_n(\textbf{x}^+)$ becomes $\mu_\beta^\pm$-a.s. constant and hence independent of the distribution of all the other coarse spins. As a direct consequence this implies that $\mathcal{T}_\alpha^n\mu_\beta^\pm$ converges weakly to a product measure. Note, however, that we cannot conclude that the spins become deterministic, only their distribution. If, for example, $f_n(\textbf{x}^+)=1/2$ $\mu_\beta^\pm$-a.s. then the corresponding coarse spin will be equally likely up or down independent on the configuration of the coarse spins elsewhere. Afterwards, we also have to determine the limit $\lim_{n\to\infty}\mu_\beta^\pm[f_n(\textbf{x}^+)]$ which characterizes the resulting product measure. This second step, however, will only be required in the low-temperature phase. In the high-temperature phase, where $\mu_\beta^+=\mu_\beta^-=\mu_\beta$ is $\mathds{Z}_2$-invariant, it is clear that for any $n\in\mathds{N}$
\begin{equation}
    \mu_\beta[f_n(\textbf{x}^+)]=\frac{1}{2}\mu_\beta[f_n(\textbf{x}^+(\sigma^{(0)}))+f_n(\textbf{x}^+(-\sigma^{(0)}))]=\frac{1}{2}\mu_\beta[f_n(\textbf{x}^+)+f_n(\textbf{x}^-)]=\frac{1}{2}.
\end{equation}

\subsubsection{High-Temperature Phase}

We want to investigate the two temperature regimes separately starting with the high-temperature phase $\beta<\beta_c$. The general idea for the low-temperature phase will be the same, although there will arise some technical difficulties we do not encounter at high temperatures. As discussed above, we want to prove that the variance $\mathrm{Var}_\beta(f_n(\textbf{x}^+))$ vanishes in the limit $n\to\infty$ which implies that $f_n(\textbf{x}^+),f_n(\textbf{x}^-)$ both attain the value $1/2$ $\mu_\beta$-a.s. In that case, we obtain
\begin{equation}
    \lim_{n\to\infty}\mu_\beta^\mathrm{f}\mathcal{T}^n[\sigma_A]=\lim_{n\to\infty}\sum_{S\subseteq A}(-1)^{|S|}\mu_\beta^\mathrm{
    f}\left[\prod_{i\in S}f_n(\textbf{x}^-_i)\prod_{j\in S^c}f_n(\textbf{x}^+_j)\right]=\left(\frac{1}{2}\right)^{|A|}\sum_{S\subseteq A}(-1)^{|S|}=\indnb{A=\emptyset}
\end{equation}
which implies convergence to the trivial $T=\infty$ fixed point.

\begin{proposition}\label{Prop2}
    At inverse temperatures $\beta<\beta_c$ and $\alpha\in(0,1/2)$ we find
    \begin{equation}
        \mathrm{Var}_\beta(f_n(\mathrm{\mathbf{x}}^+))\to0,n\to\infty.
    \end{equation}
\end{proposition}

\begin{proof}
    This proof follows an idea originally due to Efron and Stein \cite{EfronStein}. We start at fixed $n$ as well as in a fixed finite volume $\Lambda$. In the high-temperature regime the infinite volume Gibbs measure is unique such that our results will be independent of the choice of boundary conditions. Therefore, we may use $\mu_{\mathcal{C}_\Lambda^0,\beta}^\mathrm{f}$ which allows us to benefit from exact $\mathds{Z}_2$-symmetry already in finite volumes. Let us now switch to the Edwards-Sokal coupling, i.e.\ we consider
    \begin{equation}
        \mathrm{Var}_{\mathcal{C}_\Lambda^0,\beta}^\mathrm{f}(f_n)=\mu_{\mathcal{C}_\Lambda^0,\beta}^\mathrm{f}\left[\left(f_n-\mu_{\mathcal{C}_\Lambda^0,\beta}^\mathrm{f}[f_n]\right)^2\right]=\mathds{P}_{\mathcal{C}_\Lambda^0,\beta}^\mathrm{f}\left[\left(f_n-\mathds{P}_{\mathcal{C}_\Lambda^0,\beta}^\mathrm{f}[f_n]\right)^2\right],
    \end{equation}
    where we abbreviate $f_n=f_n(\textbf{x}^+)$. Using the law of total variance we write (denoting $\mathds{E}[\cdot]:=\mathds{P}_{\mathcal{C}_\Lambda^0,\beta}^\mathrm{f}[\cdot]$ and $\mathrm{Var}(\cdot):=\mathrm{Var}_{\mathcal{C}_\Lambda^0,\beta}^\mathrm{f}$ for the remainder of this proof)
    \begin{equation}
        \mathrm{Var}(f_n)=\mathds{E}[\mathrm{Var}(f_n|\omega)]+\mathrm{Var}(\mathds{E}[f_n|\omega]).
    \end{equation}
    However, since our measure is $\mathds{Z}_2$-symmetric, we can use that
    \begin{equation}
        \mathds{E}[f_n|\omega]=\frac{1}{2}\mathds{E}[f_n(\textbf{x}^+)+f_n(\textbf{x}^-)|\omega]=\frac{1}{2}
    \end{equation}
    and therefore constant. We conclude that the second term vanishes since constant functions have zero variance. For the first term we consider
    \begin{equation}
        \mathrm{Var}(f_n|\omega)=\mathds{E}[(f_n-\mathds{E}[f_n|\omega])^2|\omega].
    \end{equation}
    For clarity, let us denote $\texpec[\cdot]:=\mathds{E}[\cdot|\omega]$. Also, let $C_1,...,C_m\in\mathcal{K}(\omega)$ be the clusters that intersect $\mathcal{C}_x^0$. The measure $\texpec[\cdot]$ is simply an independent $1/2$-Bernoulli type measure over the corresponding spins $S_1,...,S_m$. Now, let us quickly compute
    \begin{equation}
        f_n-\texpec[f_n]=\texpec[f_n|S_1,...,S_m]-\texpec[f_n]=\sum_{i=1}^m\texpec[f_n|S_1,...,S_i]-\texpec[f_n|S_1,...,S_{i-1}],
    \end{equation}
    where we may denote $\Delta_if_n:=\texpec[f_n|S_1,...,S_i]-\texpec[f_n|S_1,...,S_{i-1}]$. The idea is that
    \begin{equation}
        (f_n-\texpec[f_n])^2=\sum_{i,j=1}^m\Delta_if_n\Delta_jf_n,
    \end{equation}
    where, however, the off-diagonal elements vanish in expectation. We can prove this as follows: Assume that $i<j$. Then we consider
    \begin{equation}
        \texpec[\Delta_if_n\Delta_jf_n]=\texpec[\texpec[\Delta_if_n\Delta_jf_n|S_1,...,S_{j-1}]]=\texpec[\Delta_if_n\texpec[\Delta_jf_n|S_1,...,S_{j-1}]].
    \end{equation}
    By construction we find
    \begin{equation}
        \texpec[\texpec[f_n|S_1,...,S_j]|S_1,...,S_{j-1}]=\texpec[f_n|S_1,...,S_{j-1}]\implies\texpec[\Delta_jf_n|S_1,...,S_{j-1}]=0.
    \end{equation}
    Thus, we end up with
    \begin{equation}
        \texpec[(f_n-\texpec[f_n])^2]=\sum_{i=1}^m\texpec[(\Delta_if_n)^2].
    \end{equation}
    In order to complete the proof we therefore require a bound of $|\Delta_if_n|$. Our starting point is
    \begin{align}
        \texpec[f_n|S_1,...,S_i]=\ind{S_i=+1}\texpec[f_n|S_1,...,S_{i-1},S_i=+1]+\ind{S_i=-1}\texpec[f_n|S_1,...,S_{i-1},S_i=-1].
    \end{align}
    Let us denote the terms as
    \begin{equation}
        A_i:=\texpec[f_n|S_1,...,S_{i-1},S_i=+1]\qquad;\qquad B_i:=\texpec[f_n|S_1,...,S_{i-1},S_i=-1].
    \end{equation}
    Similarly, we find
    \begin{align}
        \texpec[f_n|S_1,...,S_{i-1}]&=\texpec[\texpec[f_n|S_1,...,S]|S_1,...,S_{i-1}]\nonumber\\
        &=A_i\texpec[\ind{S_i=+1}|S_1,...,S_{i-1}]+B_i\texpec[\ind{S_i=-1}|S_1,...,S_{i-1}].
    \end{align}
    The spin variables are independent of each other such that 
    \begin{equation}
        \texpec[\ind{S_i=+1}|S_1,...,S_{i-1}]=\texpec[\ind{S_i=-1}|S_1,...,S_{i-1}]=1/2
    \end{equation}
    and therefore
    \begin{equation}
        \Delta_if_n=A_i\ind{S_i=+1}+B_i\ind{A_i=-1}-\frac{1}{2}(A_i+B_i).
    \end{equation}
    At last, we note that $\ind{S_i=-1}=1-\ind{S_i=+1}$ such that
    \begin{equation}
        \Delta_if_n=(A_i-B_i)\left(\ind{S_i=+1}-\frac{1}{2}\right).
    \end{equation}
    In absolute value, we then find $|\Delta_if_n|\leq|A_i-B_i|$. This expression can be computed via
    \begin{equation}
        A_i-B_i=\texpec[f_n-f_n^i|S_1,...,S_{i-1},S_i=+1],
    \end{equation}
    where $f_n^i$ is defined just as $f_n$ but with the sign of cluster $C_i$ flipped. All in all, we get
    \begin{align}
        \mathrm{Var}(f_n|\omega)&\leq\sum_{i=1}^m\texpec[(\texpec[|f_n-f_n^i|\,|\,S_1,...,S_{i-1},S_i=+1])^2]\nonumber\\
        &\leq\sum_{i=1}^m\texpec[\texpec[|f_n-f_n^i|^2\,|\,S_1,...,S_{i-1},S_i=+1]]=\sum_{i=1}^m\texpec[|f_n-f_n^i|^2\,|\,S_i=+1].
    \end{align}
    For the second inequality we have used Jensen's inequality \cite{Jensen}. By multilinearity of $f_n$ we can bound the difference $|f_n-f_n^i|$ by flipping the spins one by one such that
    \begin{equation}
        |f_n-f_n^i|\leq\delta f_n\cdot|C_i\cap\mathcal{C}_x^0|,
    \end{equation}
    which gives us
    \begin{equation}
        \mathrm{Var}(f_n|\omega)\leq(\delta f_n)^2\sum_{C\cap\mathcal{C}_x^0}|C\cap\mathcal{C}_x^0|^2\implies\mathds{E}[\mathrm{Var}(f_n|\omega)]\leq(\delta f_n)^2\,\mathds{E}\left[\sum_{C\cap\mathcal{C}_x^0}|C\cap\mathcal{C}_x^0|^2\right].
    \end{equation}
    The last step is to estimate the remaining expectation. This is done via
    \begin{equation}
        \sum_{C\cap\mathcal{C}_x^0}|C\cap\mathcal{C}_x^0|^2=\sum_C\left(\sum_{z\in\mathcal{C}_x^0}\ind{z\in C}\right)^2=\sum_C\sum_{z,w\in\mathcal{C}_x^0}\ind{z\in C}\ind{w\in C}.
    \end{equation}
    By changing the order of summation, we then obtain
    \begin{equation}
        \sum_{z,w\in\mathcal{C}_x^0}\sum_C\ind{z\in C}\ind{w\in C}=\sum_{z,w\in\mathcal{C}_x^0}\ind{z\leftrightarrow w}.
    \end{equation}
    Thus, we are considering
    \begin{equation}
        \mathrm{Var}(f_n)\leq(\delta f_n)^2\,\mathds{E}\left[\sum_{z,w\in\mathcal{C}_x^0}\ind{z\leftrightarrow w}\right].
    \end{equation}
    This is finally the point where we want to take the infinite volume limit $\Lambda\uparrow\mathds{Z}^2$.
    \begin{equation}
        \mathrm{Var}_\beta(f_n)\leq(\delta f_n)^2\sum_{z,w\in\mathcal{C}_x^0}\phi_p^0[z\leftrightarrow w].
    \end{equation}
    This final expression is very simple to bound as we find
    \begin{equation}
        \sum_{w\in\mathcal{C}_x^0}\phi_p^0[z\leftrightarrow w]\leq\sum_{w\in\mathds{Z}^2}\phi_p^0[z\leftrightarrow w]=\chi(\beta)<\infty.
    \end{equation}
    Putting everything together we end up with
    \begin{equation}
        \mathrm{Var}_\beta(f_n)\leq(\delta f_n)^2\chi(\beta)\cdot|\mathcal{C}_x^0|.
    \end{equation}
    Let us recall from Equation (\ref{Eq3}) that $\delta f_n=C(\alpha)^n$ where $C(\alpha)\in[1/4,1/2)$ for $\alpha\in(0,1/2)$. Since $|\mathcal{C}_x^0|=4^n$, we conclude
    \begin{equation}
        \mathrm{Var}_\beta(f_n)\to0,n\to\infty\qquad\forall\alpha\in(0,1/2).
    \end{equation}
    We have successfully established that the variance of $f_n$ and with it the correlations of $\mathcal{T}^n\mu_\beta$ vanish in the $n\to\infty$ limit for $\beta<\beta_c$.
\end{proof}

\subsubsection{Low-Temperature Phase}

Let us now consider the supercritical phase $\beta>\beta_c$. Here, obtaining the desired result requires more work than in the previous case. This is due to the fact that i) when proving that the variance vanishes we no longer profit from $\mathds{Z}_2$-symmetrie and ii) the value $\mu_\beta^\pm[f_n(\textbf{x}^+)]$ is no longer fixed by symmetrie arguments. Hence, alongside proving that the variance vanishes we must also find a way to determine the value $\lim_{n\to\infty}\mu_\beta^\pm[f_n(\textbf{x}^+)]$. Let us for now focus on minus boundary conditions $\mu_\beta^-$. Similar to the high-temperature phase, the first step is again to prove that the variance of $f_n$ vanishes. 

\begin{proposition}\label{Prop1}
    For $\beta>\beta_c$ the variance vanishes for $\alpha\in(0,1/2)$ in the limit
    \begin{equation}
        \lim_{n\to\infty}\mathrm{Var}^-_\beta(f_n(\mathrm{\mathbf{x}}^+))=0.
    \end{equation}
\end{proposition}

\begin{proof}
    The beginning of this proof is very similar to the previous one for Proposition \ref{Prop2}. Switching to the Edwards-Sokal coupling and then applying the law of total variance we again consider the expression
    \begin{equation}
        \mathrm{Var}(f_n(\textbf{x}^+))=\mathds{E}[\mathrm{Var}(f_n(\textbf{x}^+)|\omega)]+\mathrm{Var}(\mathds{E}[f_n(\textbf{x}^+)|\omega]).
    \end{equation}
    Note, however, that this time we consider $-$-boundary conditions, i.e.\ we set $\mathds{E}[\cdot]=\mathds{P}_{\mathcal{C}_\Lambda^0,\beta}^-$ and $\mathrm{Var}(\cdot)=\mathrm{Var}_{\mathcal{C}_\Lambda^0,\beta}^-$. One rather unfortunate consequence is that the second term $\mathrm{Var}(\mathds{E}[f_n(\textbf{x}^+)|\omega])$ no longer vanishes for finite $n$. This is a by-product of us losing the $\mathds{Z}_2$-symmetry such that $\mathds{E}[f_n^x(\textbf{x}^+)|\omega]$ now depends on the bond configuration $\omega$, particularly the shape of $C_\infty$ which denotes the cluster containing the boundary of our finite volume (it will become the unique infinitely large cluster in the infinite volume limit). The good news are that obtaining a bound for the first term is almost equal to the proof of Proposition \ref{Prop2}. All steps up to the point where we took the infinite volume limit remain exactly the same with one exception. The boundary cluster $C_\infty$ has a fixed spin and thus does not contribute to the variance. Therefore, we find
    \begin{equation}
        \mathrm{Var}(f_n|\omega)\leq(\delta f_n)^2\sum_{C\in\mathcal{K}(\omega)\backslash \{C_\infty\}}|C\cap\mathcal{C}_x^0|^2.
    \end{equation}
    Here, we may rewrite
    \begin{equation}
        \sum_{C\in\mathcal{K}(\omega)\backslash \{C_\infty\}}|C\cap\mathcal{C}_x^0|^2=\sum_{z,w\in\mathcal{C}_x^0}\ind{z\notin C_\infty,z\leftrightarrow w},
    \end{equation}
    which gives us
    \begin{equation}\label{Eq6}
        \mathds{E}[\mathrm{Var}(f_n|\omega)]\leq(\delta f_n)^2\sum_{z,w\in\mathcal{C}_x^0}\mathds{E}[\ind{z\notin C_\infty,z\leftrightarrow w}].
    \end{equation}
    Before we take the infinite volume limit, we also want to find a bound for the second term. For better readability, let us denote $F(\omega):=\mathds{E}[f_n^x(\textbf{x}^+)|\omega]$. The plan is to once again use an idea akin to the Efron-Stein inequality. However, now our random variables are the bonds $\omega_e$ which are no longer independent making things a bit more tricky. We start off by implementing the martingale decomposition
    \begin{equation}
        F-\mathds{E}[F]=\sum_{i=1}^m\mathds{E}[F|\omega_1,...,\omega_i]-\mathds{E}[F|\omega_1,...,\omega_{i-1}],
    \end{equation}
    where we have ordered the edges $E(\mathcal{C}_\Lambda^0)=\{e_1,...,e_m\}$ in some way. The fact that in expectation, the off-diagonal terms vanish has nothing to do with the question of (in-)dependence of the variables. Therefore, we still obtain
    \begin{equation}\label{Eq5}
        \mathrm{Var}(F)=\sum_{i=1}^m\mathds{E}[(\Delta_iF)^2].
    \end{equation}
    The subsequent step we made in the proof of Proposition \ref{Prop2} can also be applied here, i.e.\ we find $|\Delta_iF|\leq|A_i-B_i|$ where
    \begin{equation}
        A_i,B_i=\mathds{E}[F|\omega_1,...,\omega_{i-1},\omega_i=1,0].
    \end{equation}
    The main complication due to the now dependent bond variables $\omega_i$ arises in the computation of $|A_i-B_i|$. In the independent case we were able to do something like $\mathds{E}[f]-\hat{\mathds{E}}[f]=\mathds{E}[f-\hat f]$ and then bound the difference $|f-\hat f|$. In order to achieve something similar here, we will prove in Lemma \ref{Lem4} below the existence of a coupling $\mathds{P}_i(\omega,\omega')$ with the properties
    \begin{enumerate}
        \item[i)] $\mathds{P}_i[F(\omega)-F(\omega')]=A_i-B_i$,
        \item[ii)] $\mathds{P}_i[\omega'\leq\omega]=1$,
        \item[iii)] $\mathds{P}_i[D(\omega',\omega)^*\subseteq C_{e_i^*}(\omega'^*)]=1$ where $D(\omega',\omega):=\{e\in\mathcal{E}\,:\,\omega_e'<\omega_e\}$.
    \end{enumerate}
    Recall that $\mathcal{E}^*$ is the set of dual edges and $\omega'^*$ is the configuration dual to $\omega'$. For now, let us continue under the assumption that such a coupling exists. Also, we may denote $\mathcal{K}^\circ(\omega):=\mathcal{K}(\omega)\backslash\{C_\infty\}$. In order to find a bound on $|A_i-B_i|$ we should be looking into
    \begin{equation}
        |F(\omega)-F(\omega')|=\left|\frac{1}{2^{|\mathcal{K}^\circ(\omega)|}}\sum_{\sigma\in\{-1,1\}^{\mathcal{K}^\circ(\omega)}}f_n(\textbf{x}^+(\sigma))-\frac{1}{2^{|\mathcal{K}^\circ(\omega')|}}\sum_{\sigma'\in\{-1,1\}^{\mathcal{K}^\circ(\omega')}}f_n(\textbf{x}^+(\sigma'))\right|.
    \end{equation}
    We may consider an embedding $i:\mathcal{K}(\omega)\to\mathcal{K}(\omega')$ such that $i(C):=\max\{|C'|\,:\,C'\in\mathcal{K}(\omega'),C'\subseteq C\}$. Note that, due to $\omega'\leq\omega$, every cluster $C\in\mathcal{K}(\omega)$ decomposes into a union of clusters in $\mathcal{K}(\omega')$. $i(C)$ then simply selects the largest constituent. Doing so, we find
    \begin{equation}
        F(\omega)-F(\omega')=\frac{1}{2^{|\mathcal{K}^\circ(\omega')|}}\sum_{\sigma\in\{-1,1\}^{\mathcal{K}^\circ(\omega)}}\sum_{\tilde\sigma\in\{-1,1\}^{\mathcal{K}^\circ(\omega')\backslash i(\mathcal{K}^\circ(\omega))}}\left[f_n(\textbf{x}^+(\sigma))-f_n(\textbf{x}^+(s(\sigma,\tilde\sigma)))\right]
    \end{equation}
    where $s(\sigma,\tilde\sigma)\in\{-1,1\}^{\mathcal{K}^\circ(\omega')}$ is a spin configuration defined via
    \begin{equation}
        s_{C'}(\sigma,\tilde\sigma):=\begin{cases}
            \sigma_{i^{-1}(C')}&C'\in i(\mathcal{K}^\circ(\omega))\\
            \tilde\sigma_{C'}&C'\notin i(\mathcal{K}^\circ(\omega))
        \end{cases}.
    \end{equation}
    Therefore, once again using the bound $\delta f_n$ together with the multilinearity of $f_n$, we obtain
    \begin{equation}
        |F(\omega)-F(\omega')|\leq\delta f_n\underbrace{\bigcup_{C'\in\mathcal{K}^\circ(\omega')\backslash i(\mathcal{K}^\circ(\omega))}|C'|}_{=:h(\omega,\omega')}.
    \end{equation}
    Inserting this bound into the computation of $|A_i-B_i|$, we arrive at
    \begin{equation}
        \mathrm{Var}(F)\leq(\delta f_n)^2\sum_{i=1}^m\mathds{E}\left[\left(\mathds{P}_i[h(\omega,\omega')]\right)^2\right]\leq(\delta f_n)^2\sum_{i=1}^m\mathds{E}\left[\mathds{P}_i[h(\omega,\omega')^2]\right],
    \end{equation}
    where we have again made use of Jensen's inequality in order to pull the square inside the expectation. Now let us think about the ramifications of property iii) of the coupling $\mathds{P}_i$ on the value of $h(\omega,\omega')$. Because $\omega,\omega'$ use wired boundary conditions, their duals have free boundary conditions. Therefore, $C_{e_i^*}(\omega'^*)$ is surely bounded. Since the cluster structures of $\omega,\omega'$ agree outside of $C_{e_i^*}(\omega'^*)$, it is always possible to choose an embedding that hits all spins not cut off from the universe through $C_{e_i^*}(\omega'^*)$ (although this might not be the one that minimizes $h$). See Figure \ref{Fig2} for a visualization. Therefore, we find that
    \begin{equation}
        h(\omega,\omega')\leq\sum_{z\in\mathcal{C}_x^0}\ind{C_{e_i^*}(\omega'^*)\circlearrowleft z}\implies\mathds{P}_i[h(\omega,\omega')^2]\leq\sum_{z,w\in\mathcal{C}_x^0}\mathds{P}_i[\ind{C_{e_i^*}(\omega'^*)\circlearrowleft z,C_{e_i^*}(\omega'^*)\circlearrowleft w}],
    \end{equation}
    where $C_{e^*}(\omega'^*)\circlearrowleft z$ corresponds to the event that the spin $z$ is completely surrounded by $C_{e^*}(\omega'^*)$. The remaining indicator function over the event $\mathcal{A}_{z,w}=\{C_{e_i^*}(\omega'^*)\circlearrowleft z,C_{e_i^*}(\omega'^*)\circlearrowleft w\}$ only depends on $\omega'$ which gives us
    \begin{equation}
        \sum_{z,w\in\mathcal{C}_x^0}\mathds{E}\left[\mathds{P}_i[\mathcal{A}_{z,w}]\right]=\sum_{z,w\in\mathcal{C}_x^0}\mathds{E}\left[\mathds{E}[\mathcal{A}_{z,w}|\omega_1,...,\omega_{i-1},\omega_i=0]\right]=\sum_{z,w\in\mathcal{C}_x^0}\mathds{E}[\mathcal{A}_{z,w}|\omega_i=0].
    \end{equation}
    Putting everything together, we are looking at
    \begin{equation}\label{Eq8}
        \mathrm{Var}(F)\leq(\delta f_n)^2\sum_{e\in E(\mathcal{C}_\Lambda^0)}\sum_{z,w\in\mathcal{C}_x^0}\mathds{E}[\mathcal{A}_{z,w}|\omega_e=0].
    \end{equation}
    Together with the result from (\ref{Eq6}) we arrive at the total bound
    \begin{equation}
        \mathrm{Var}(f_n)\leq(\delta f_n)^2\left(\sum_{z,w\in\mathcal{C}_x^0}\mathds{E}[\ind{z\notin C_\infty,z\leftrightarrow w}]+\sum_{e\in E(\mathcal{C}_\Lambda^0)}\sum_{z,w\in\mathcal{C}_x^0}\mathds{E}[\mathcal{A}_{z,w}|\omega_e=0]\right).
    \end{equation}
    This is the appropriate point to take the infinite volume limit. In order to do so we note that
    \begin{equation}
        \mathds{E}[\mathcal{A}_{z,w}|\omega_e=0]=\frac{\mathds{E}[\mathcal{A}_{z,w}]}{\mathds{E}[\omega_e=0]}.
    \end{equation}
    We do not have to worry about false contributions from configurations with $\omega_e=1$ as in that case $C_{e^*}(\omega^*)=\emptyset$ and thus such configurations are not in $\mathcal{A}_{z,w}$. By translation invariance of the infinite volume measure we note that $\mathds{E}[\omega_e=0]$ becomes independent of the edge $e$ in the limit $\Lambda\uparrow\mathds{Z}^2$ while it also remains bounded away from zero. This gives us
    \begin{equation}\label{Eq9}
        \mathrm{Var}_\beta^-(f_n)\leq(\delta f_n)^2\left(\sum_{z,w\in\mathcal{C}_x}\phi_{p(\beta)}[z\nlr\infty,z\leftrightarrow w]+\frac{1}{\phi_{p(\beta)}[\omega_0=0]}\sum_{e\in\mathcal{E}}\sum_{z,w\in\mathcal{C}_x^0}\phi_{p(\beta)}[\mathcal{A}_{z,w}]\right).
    \end{equation}

    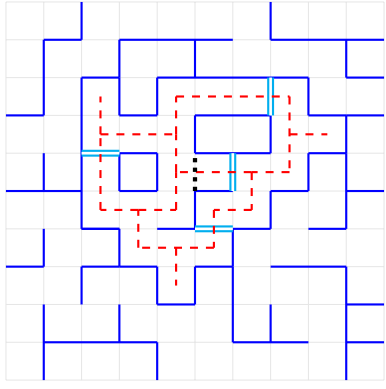
\begin{figure}[h!]
        \centering
        \begin{tikzpicture}
          \draw[step=0.5cm,gray!20,very thin] (-2.5,-2.5) grid (2.5,2.5);
          
          \draw[thick,blue] (0,0.5) -- (0,1);
          \draw[thick,blue] (0,1) -- (1,1);
          \draw[thick,blue] (0,0.5) -- (1,0.5);
          \draw[thick,blue] (1,0.5) -- (1,1);
          
          \draw[thick,blue] (0,0) -- (0.5,0);
          \draw[thick,blue] (0,0) -- (0,-0.5);
          \draw[thick,blue] (-0.5,0) -- (-1,0);
          \draw[thick,blue] (0,-0.5) -- (-0.5,-0.5);
          \draw[thick,blue] (-0.5,0) -- (-0.5,0.5);
          \draw[thick,blue] (-0.5,0.5) -- (-1,0.5);
          \draw[thick,blue] (-1,0.5) -- (-1,0);
    
          \draw[thick,blue] (-1.5,-0.5) -- (-1.5,1.5);
          \draw[thick,blue] (-1.5,1.5) -- (-1,1.5);
          \draw[thick,blue] (-1,1.5) -- (-1,1);
          \draw[thick,blue] (-1,1) -- (-0.5,1);
          \draw[thick,blue] (-0.5,1) -- (-0.5,1.5);
          \draw[thick,blue] (-0.5,1.5) -- (1.5,1.5);
          \draw[thick,blue] (1.5,1.5) -- (1.5,1);
          \draw[thick,blue] (1.5,1) -- (2,1);
          \draw[thick,blue] (2,1) -- (2,0.5);
          \draw[thick,blue] (2,0.5) -- (1.5,0.5);
          \draw[thick,blue] (1.5,0.5) -- (1.5,0);
          \draw[thick,blue] (1.5,0) -- (1,0);
          \draw[thick,blue] (1,0) -- (1,-0.5);
          \draw[thick,blue] (1,-0.5) -- (0.5,-0.5);
          \draw[thick,blue] (0.5,-0.5) -- (0.5,-1);
          \draw[thick,blue] (0.5,-1) -- (0,-1);
          \draw[thick,blue] (0,-1) -- (0,-1.5);
          \draw[thick,blue] (0,-1.5) -- (-0.5,-1.5);
          \draw[thick,blue] (-0.5,-1.5) -- (-0.5,-1);
          \draw[thick,blue] (-0.5,-1) -- (-1,-1);
          \draw[thick,blue] (-1,-1) -- (-1,-0.5);
          \draw[thick,blue] (-1,-0.5) -- (-1.5,-0.5);

          \draw[thick,blue] (-1,-1) -- (-1.5,-1);
          \draw[thick,blue] (-1.5,-1) -- (-1.5,-1.5);
          \draw[thick,blue] (-2.5,-1) -- (-2,-1);
          \draw[thick,blue] (-2,-0.5) -- (-2,-1);
          \draw[thick,blue] (-1.5,0) -- (-2.5,0);
          \draw[thick,blue] (-2,0) -- (-2,0.5);
          \draw[thick,blue] (-2.5,1) -- (-2,1);
          \draw[thick,blue] (-2,1) -- (-2,2);
          \draw[thick,blue] (-2,2) -- (-1.5,2);
          \draw[thick,blue] (-1.5,2) -- (-1.5,2.5);
          \draw[thick,blue] (-1,1.5) -- (-1,2);
          \draw[thick,blue] (-1,2) -- (0.5,2);
          \draw[thick,blue] (0,1.5) -- (0,2);
          \draw[thick,blue] (1,2.5) -- (1,2);
          \draw[thick,blue] (1,2) -- (2.5,2);
          \draw[thick,blue] (2,2) -- (2,1.5);
          \draw[thick,blue] (2,1.5) -- (2.5,1.5);
          \draw[thick,blue] (-1,-0.5) -- (-1.5,-0.5);
          \draw[thick,blue] (2,0.5) -- (2,0);
          \draw[thick,blue] (2,0) -- (2.5,0);
          \draw[thick,blue] (2,0) -- (2,-0.5);
          \draw[thick,blue] (2,-0.5) -- (1.5,-0.5);
          \draw[thick,blue] (2,1) -- (2.5,1);
          \draw[thick,blue] (0.5,-1) -- (0.5,-2);
          \draw[thick,blue] (0.5,-2) -- (1.5,-2);
          \draw[thick,blue] (1,-2) -- (1,-1.5);
          \draw[thick,blue] (1,-1) -- (2,-1);
          \draw[thick,blue] (2,-1) -- (2,-2.5);
          \draw[thick,blue] (2,-1.5) -- (2.5,-1.5);
          \draw[thick,blue] (2,-2) -- (2.5,-2);
          \draw[thick,blue] (-2,-1.5) -- (-2,-2.5);
          \draw[thick,blue] (-2,-2) -- (-0.5,-2);
          \draw[thick,blue] (-1,-2) -- (-1,-1.5);
          \draw[thick,blue] (-0.5,-2) -- (-0.5,-2.5);

          \draw[double, line width=0.8pt, double distance=1pt, cyan] (0.5,0) -- (0.5,0.5);
          \draw[double, line width=0.8pt, double distance=1pt, cyan] (-1,0.5) -- (-1.5,0.5);
          \draw[double, line width=0.8pt, double distance=1pt, cyan] (0,-0.5) -- (0.5,-0.5);
          \draw[double, line width=0.8pt, double distance=1pt, cyan] (1,1) -- (1,1.5);
          
          \draw[thick, red, dashed] (-0.75,-0.75) -- (-0.75,-0.25);
          \draw[thick, red, dashed] (-0.75,-0.25) -- (-1.25,-0.25);
          \draw[thick, red, dashed] (-1.25,-0.25) -- (-1.25,0.75);
          \draw[thick, red, dashed] (-1.25,0.75) -- (-0.25,0.75);
          \draw[thick, red, dashed] (-0.25,0.75) -- (-0.25,1.25);
          \draw[thick, red, dashed] (-0.25,1.25) -- (1.25,1.25);
          \draw[thick, red, dashed] (1.25,1.25) -- (1.25,0.25);
          \draw[thick, red, dashed] (1.25,0.25) -- (0.75,0.25);
          \draw[thick, red, dashed] (0.75,0.25) -- (0.75,-0.25);
          \draw[thick, red, dashed] (0.75,-0.25) -- (0.25,-0.25);
          \draw[thick, red, dashed] (0.25,-0.25) -- (0.25,-0.75);
          \draw[thick, red, dashed] (0.25,-0.75) -- (-0.75,-0.75);
    
          \draw[thick, red, dashed] (0.75,0.25) -- (-0.25,0.25);
          \draw[thick, red, dashed] (-0.25,0.25) -- (-0.25,0.75);
          \draw[thick, red, dashed] (-0.25,-0.75) -- (-0.25,-1.25);
          \draw[thick, red, dashed] (1.25,0.75) -- (1.75,0.75);
          \draw[thick, red, dashed] (-1.25,0.75) -- (-1.25,1.25);
          \draw[thick, red, dashed] (-0.75,-0.25) -- (-0.25,-0.25);
          \draw[thick, red, dashed] (-0.25,-0.25) -- (-0.25,0.25);

          \draw[ultra thick, black, dotted] (0,0) -- (0,0.5);
          
          
        \end{tikzpicture}
    
        \caption{The solid edges represent a possible configuration $\omega'$ with the dotted edge corresponding to $e_i$, which is set to $0$ for $\omega'$ and $1$ for $\omega$. The dashed edges represent the dual cluster $C_{e_i^*}(\omega'^*)$. At last, the double edges are additional open edges in $\omega$ that are closed in $\omega'$. It is clear that spin configurations on $\omega$ and $\omega'$ can always be paired (up to multiplicity) such that they match everywhere except on the few islands inside of $C_{e_i^*}(\omega'^*)$. In $\omega$, these islands may connect to one another and therefore must have a common spin while they remain separated in $\omega'$ and possess independent spins.}
        \label{Fig2}
    \end{figure}

    The remaining step is to prove that both sums in (\ref{Eq9}) are of order $O(|\mathcal{C}_x^0|)$ and thus get suppressed by $(\delta f_n)^2$ in the limit $n\to\infty$. Recall from the introduction that in the supercritical phase, the susceptibility of finite clusters is finite. Thus, we find the rather simple bound
    \begin{equation}
        \sum_{z,w\in\mathcal{C}_x}\phi_{p(\beta)}[z\nlr\infty,z\leftrightarrow w]\leq|\mathcal{C}_x^0|\cdot\chi_p^f
    \end{equation}
    for the first sum. In order to obtain a bound for the second sum, we want to start by explicitly carrying out the sum over $w$ giving us
    \begin{align}
        \sum_{w\in\mathcal{C}_0^x}\phi_{p(\beta)}[C_{e^*}(\omega^*)\circlearrowleft z,C_{e^*}(\omega^*)\circlearrowleft w]&\leq\sum_{w\in\mathds{Z}^2}\phi_{p(\beta)}[C_{e^*}(\omega^*)\circlearrowleft z,C_{e^*}(\omega^*)\circlearrowleft w]\nonumber\\
        &=\phi_{p(\beta)}[\ind{C_{e^*}(\omega^*)\circlearrowleft z}\cdot|\mathrm{Int}(C_{e^*}(\omega^*)|],
    \end{align}
    where we set the interior of a dual cluster to $\mathrm{Int}(C^*):=\{y\in\mathds{Z}^2\,:\,C^*\circlearrowleft y\}$. The isoperimetric inequality on $\mathds{Z}^2$ then gives us
    \begin{equation}
        |\mathrm{Int}(C^*)|\leq \frac{1}{16}|C^*|^2
    \end{equation}
    and thus
    \begin{equation}
        \sum_{e\in\mathcal{E}}\sum_{z,w\in\mathcal{C}_x^0}\phi_{p(\beta)}[\mathcal{A}_{z,w}]\leq\frac{1}{16}\sum_{z\in\mathcal{C}_x^0}\sum_{e\in\mathcal{E}}\phi_{p(\beta)}[\ind{C_{e^*}(\omega^*)\circlearrowleft z}\cdot|C_{e^*}(\omega^*)|^2].
    \end{equation}
    Using Cauchy-Schwarz we further obtain
    \begin{equation}
        \phi_{p(\beta)}[\ind{C_{e^*}(\omega^*)\circlearrowleft z}\cdot|C_{e^*}(\omega^*)|]\leq\sqrt{\phi_{p(\beta)}[\ind{C_{e^*}(\omega^*)\circlearrowleft z}]}\sqrt{\phi_{p(\beta)}[|C_{e^*}(\omega^*)|^4]}.
    \end{equation}
    Since $p(\beta)>p_c$, the dual system is subcritical. Therefore, Lemma \ref{Lem1} can be used to deduce $\sqrt{\phi_{p(\beta)}[|C_{e^*}(\omega^*)|^4]}=M<\infty$. Also, by translation invariance $M$ is independent of the choice of $e\in\mathcal{E}$. So far, we managed to arrive at
    \begin{equation}
        \sum_{e\in\mathcal{E}}\sum_{z,w\in\mathcal{C}_x^0}\phi_{p(\beta)}[\mathcal{A}_{z,w}]\leq\frac{M}{16}\sum_{z\in\mathcal{C}_x^0}\sum_{e\in\mathcal{E}}\sqrt{\phi_{p(\beta)}[\ind{C_{e^*}(\omega^*)\circlearrowleft z}]}=\frac{M}{16}|\mathcal{C}_x^0|\sum_{e\in\mathcal{E}}\sqrt{\phi_{p(\beta)}[\ind{C_{e^*}(\omega^*)\circlearrowleft 0}]},
    \end{equation}
    where we again utilized the translation invariance of $\phi_{p(\beta)}$. The final remaining step is to use
    \begin{equation}
        \sqrt{\phi_{p(\beta)}[\ind{C_{e^*}(\omega^*)\circlearrowleft 0}]}\leq\sqrt{\phi_{p(\beta)}[\exists y^*\,:\,e^*\leftrightarrow y^*,\|y^*-e^*\|>\|e^*\|]}\leq\sqrt{C_{p^*}'}e^{-c_{p^*}'\|e^*\|/2},
    \end{equation}
    where we have again leveraged the fact that the dual system is subcritical. The distance from the origin $\|e^*\|$ is simply the minimal distance over the two endpoints of $e^*$. This exponential decay implies that the sum over $e\in\mathcal{E}$ is finite completing the proof.
\end{proof}

In order to finalize the previous proof we now want to establish the existence of the coupling $\mathds{P}$ we made use of. This proof is based on the proof of Lemma 1.5 in \cite{10.1007/978-3-030-32011-9_2}. The only difference is the addition of property iii).

\begin{lemma}\label{Lem4}
    Let $G=(V,E)\subset\mathds{Z}^2$ be some subgraph, for once not necessarily connected, induced and with a connected complement\footnote{We want to apply this lemma to the subgraph $\mathcal{C}_\Lambda^0$ with edges $e_1,...,e_i$ removed which does not enjoy these properties even if $\mathcal{C}_\Lambda^0$ does.}. Further, let $\eta,\tilde\eta\in\Omega_{\mathcal{E}\backslash E}$ be some planar boundary conditions such that
    \begin{equation}
        \eta_e=\tilde\eta_e\quad\forall e\in\mathcal{E}\backslash(E\cup\{a\})
    \end{equation}
    for some fixed edge $a\in\mathcal{E}\backslash E$. At $a$, we fix $\eta_a=0,\tilde\eta_a=1$. Then there exists a coupling $\mathds{P}(\omega,\tilde\omega)$ on $\{0,1\}^{E}$ satisfying
    \begin{enumerate}
        \item[i)] $\mathds{P}(\omega)=\phi_{p,G}^\eta(\omega),\mathds{P}(\tilde\omega)=\phi_{p,G}^{\tilde\eta}(\tilde\omega)$,
        \item[ii)] $\mathds{P}[\omega\leq\tilde\omega]=1$,
        \item[iii)] $\mathds{P}[D(\omega,\tilde\omega)^*\subseteq C_{a^*}(\omega^*,\eta^*)]=1$ where $D(\omega,\tilde\omega):=\{e\in E\,:\,\omega_e<\tilde\omega_e\}$.
    \end{enumerate}
\end{lemma}

However, for better readibility, let us separate the following lemma from the main proof.

\begin{lemma}\label{Lem2}
    Let $\omega,\tilde\omega\in\Omega_\mathcal{E}$ such that $\omega\leq\tilde\omega$ and there exists some cluster $C\in\mathcal{K}(\omega^*)$ such that $\omega_e<\tilde\omega_e\iff e^*\in C$. Then we find that for any edge $\braket{xy}\in\mathcal{E}$ such that $\omega_{\braket{xy}}=0$ and $\braket{xy}^*\notin C$ that
    \begin{equation}
        x\overset{\omega}{\longleftrightarrow}y\iff x\overset{\tilde\omega}{\longleftrightarrow}y.
    \end{equation}
\end{lemma}

\begin{proof}
    Due to $\omega\leq\tilde\omega$ the $\implies$-direction is trivial. For the converse, let $P=[v_0=x,v_1,...,,v_N=y]$ be a $\tilde\omega$-open path that contains some edges which are closed in $\omega$.  Further, let $i=\min\{k\in\{0,...,N\}\,:\,\braket{v_kv_{k+1}}^*\in C\}$ and $j=\max\{k\in\{0,...,N\}\,:\,\braket{v_{k-1}v_k}^*\in C\}$. Since $P$ contains at least one edge that is closed in $\omega$, implying that its dual lies in $C$, we deduce that $i<j$ are well-defined. By construction, all edges $\braket{v_{l-1}v_l},l=1,...,i$ are open in $\tilde\omega$ and their duals are not contained in $C$, hence, they must be open in $\omega$ as well. We conclude that $x\overset{\omega}{\longleftrightarrow}v_i$ and by the same line of reasoning also $y\overset{\omega}{\longleftrightarrow}v_j$ (in case that $i=0$ or $j=N$ the connection is trivial). Furthermore, $v_i,v_j$ lie on the same $\omega$-open boundary of $C$ as $x,y$ must be part of the same connected component of the complement of $C$ (due to the condition that $\braket{xy}^*\notin C$). This implies that $x\overset{\omega}{\longleftrightarrow}v_i\overset{\omega}{\longleftrightarrow}v_j\overset{\omega}{\longleftrightarrow}y$. We conclude that any $\tilde\omega$-open path that connects $x,y$ and is closed in $\omega$ allows for a construction of an $\omega$-open path connecting $x,y$ which proves $\impliedby$.
\end{proof}

With this in mind, we shall proceed with the proof of Lemma \ref{Lem4}.\medskip

\textit{Proof of Lemma \ref{Lem4}.} Our goal is to construct a continuous-time Markov chain which converges to the desired coupling. Therefore, let us assign both an exponential clock\footnote{We can think of an exponential clock as some large number of unstable atoms where the clock rings whenever an atom decays. However, every time one atom decays we add a new one such that the average decay rate remains the same for all times.} $\mathscr{C}_e$ as well as a family of uniform distributed random variables $U_{e,k}\in[0,1],k\in\mathds{N}$ to every edge $e\in E$. Further, we choose some arbitrary starting configuration $\omega^0=\tilde\omega^0=\omega_0$. The evolution in time works as follow: If at time $t$ the clock $\mathscr{C}_e$ rings for the $k$th time, we update the bond at $e$ via
\begin{align}
    \omega^t_e:=\begin{cases}
        1&U_{e,k}\geq\phi_{p,G}^\eta[\omega_e=0\,|\,\omega_{E\backslash \{e\}}=\omega_{E\backslash \{e\}}^{t^-}]\\
        0&\mathrm{otherwise}
    \end{cases},\nonumber\\
    \tilde\omega^t_e:=\begin{cases}
        1&U_{e,k}\geq\phi_{p,G}^{\tilde\eta}[\omega_e=0\,|\,\omega_{E\backslash \{e\}}=\tilde\omega_{E\backslash \{e\}}^{t^-}]\\
        0&\mathrm{otherwise}
    \end{cases},
\end{align}
where $\omega^{t^-},\tilde\omega^{t^-}$ are the configurations just before the clock rang. Note that we assume $p\in(0,1)$ here such that the Markov chains $(\omega^t)_{t\in\mathds{R}_{\geq0}}$ and $(\tilde\omega^t)_{t\in\mathds{R}_{\geq0}}$ are both irreducible. Further, the rules for the bond updates are chosen such that $\phi_{p,G}^\eta$ and $\phi_{p,G}^{\tilde\eta}$ are the (by irreducibility of the chains unique) stationary measures. Therefore, the law of $\omega^t$ converges to $\phi_{p,G}^\eta$ and likewise for $\tilde\omega^t$ and $\phi_{p,G}^{\tilde\eta}$. Properties ii) and iii) can be proven using a continuous time induction. In order to do so, we may assume that at time $t$ the clock $\mathscr{C}_{\braket{xy}}$ rings and the pair $(\omega^{t^-},\tilde\omega^{t^-})$ satisfies both ii) and iii) (both properties are clear at $t=0$). First, let us quickly compute the conditional expectation $\phi_{p,G}^\eta[\omega_{\braket{xy}}=0\,|\,\omega_{E\backslash \{\braket{xy}\}}=\omega_{E\backslash \{\braket{xy}\}}^{t^-}]$. In case that $x,y$ are already connected with edges in $\mathcal{E}\backslash\{\braket{xy}\}$, we find
\begin{equation}
    \phi_{p,G}^\eta[\omega_{\braket{xy}}=0\,|\,\omega_{E\backslash \{\braket{xy}\}}=\omega_{E\backslash \{\braket{xy}\}}^{t^-}]=1-p.
\end{equation}
If, however, the edge $\braket{xy}$ would connect two previously disjoint clusters we would instead find
\begin{equation}
    \phi_{p,G}^\eta[\omega_{\braket{xy}}=0\,|\,\omega_{E\backslash \{\braket{xy}\}}=\omega_{E\backslash \{\braket{xy}\}}^{t^-}]=\frac{2(1-p)}{2(1-p)+p}=\frac{1-p}{1-p/2}\geq1-p.
\end{equation}
For ii), we realize that $\omega^{t^-}\leq\tilde\omega^{t^-},\eta\leq\tilde\eta\implies$ whenever $x,y$ are connected in $\mathcal{E}\backslash\{\braket{xy}\}$ regarding $(\omega^{t^-},\eta)$ they are so regarding $(\tilde\omega^{t^-},\tilde\eta)$ as well. Thus, we obtain
\begin{equation}
    \phi_{p,G}^\eta[\omega_{\braket{xy}}=0\,|\,\omega_{E\backslash \{\braket{xy}\}}=\omega_{E\backslash \{\braket{xy}\}}^{t^-}]\geq\phi_{p,G}^{\tilde\eta}[\omega_e=0\,|\,\omega_{E\backslash \{e\}}=\tilde\omega_{E\backslash \{e\}}^{t^-}]\implies\omega^t\leq\tilde\omega^t.
\end{equation}
Therefore, ii) still holds at time $t$ and thus by induction for all times. For iii), we shall additionally assume that $\braket{xy}^*\not\leftrightarrow C_{a^*}((\omega^{t^-})^*,\eta^*)$ (meaning neither of the two endpoints of $\braket{xy}^*$ are contained in $C_{a^*}((\omega^{t^-})^*,\eta^*)$). Let us set the two configurations $\xi,\tilde\xi\in\Omega_\mathcal{E}$ defined as
\begin{equation}
    \xi_e:=\begin{cases}
        \omega^{t^-}_e&e\in E\backslash\{\braket{xy}\}\\
        0&e=\braket{xy}\\
        \eta_e&e\in\mathcal{E}\backslash E
    \end{cases},\qquad\tilde\xi_e:=\begin{cases}
        \tilde\omega^{t^-}_e&e\in E\backslash\{\braket{xy}\}\\
        0&e=\braket{xy}\\
        \tilde\eta_e&e\in\mathcal{E}\backslash E
    \end{cases}.
\end{equation}
By our induction assumption, the pairing $\xi,\tilde\xi$ satisfy the conditions of Lemma \ref{Lem2} where the cluster containing the differences is $C_{a^*}((\omega^{t^-})^*,\eta^*)$. Since $\braket{xy}^*\notin C_{a^*}((\omega^{t^-})^*,\eta^*)$, we conclude that $x\overset{\xi}{\longleftrightarrow}y\iff x\overset{\tilde\xi}{\longleftrightarrow}y$. This implies that the bond gets updated in the same manner for both $\omega^t$ and $\tilde\omega^t$ and hence, possible differences can only be generated at edges for which their dual connects to $C_{a^*}((\omega^{t^-})^*,\eta^*)$.\qed

\medskip

This lemma completes the proof of Proposition \ref{Prop1}. Thus, we have now successfully established that for any $\beta>\beta_c$ the variance of $f_n(\textbf{x}^+)$ regarding $\mu_{\beta}^-$ vanishes in the limit $n\to\infty$. By symmetry of the problem, this of course also proves the same statement for $\mu_\beta^+$. However, recall that it is not clear so far which value $f_n(\textbf{x}^+)$ attains in the limit. In the high-temperature phase we knew that the system remained $\mathds{Z}_2$-symmetric and thus $f_n(\textbf{x}^+)$ had to attain the value $1/2$ $\mu_\beta$-a.s. In the symmetry broken phase the vanishing variance only tells us that $f_n(\textbf{x}^+)$ attains \textbf{some} value $\mu_\beta^-$-a.s. in the limit $n\to\infty$. In order to figure out which value, we consider the recursion
\begin{equation}
    f_n(\textbf{x}^+)=f(f_{n-1}(\textbf{x}^+_1),...,f_{n-1}(\textbf{x}^+_4)).
\end{equation}
In the limit $n\to\infty$ the $f_{n-1}(\textbf{x}^+_i)$ become $\mu_\beta^-$-a.s. constant and by translation invariance, they also attain the same value. Therefore, we must investigate the dynamics of the map
\begin{equation}
    g^\alpha(x):=f^\alpha(x,...,x).
\end{equation}
$f^\alpha$ corresponds to the evolution map $f=f_1$ evaluated at the current value of $\alpha$. The behaviour is very different depending on whether we are in the regime $\alpha<1/4,\,\alpha=1/4$ or $\alpha>1/4$. Figure \ref{Fig4} presents a qualitative picture of the dynamics of $g^\alpha$. For $\alpha<1/4$, $g^\alpha$ has an unstable fixed point at $x=1/2$ while everything else gets attracted by the stable fixed points at $x=0,1$. At $\alpha=1/4$, $g^\alpha$ is simply the identity and then for $\alpha>1/4$ the behaviour reverses, i.e.\ $x=1/2$ becomes stable and attracts everything besides the two unstable fixed points $x=0,1$. We would like to conclude from this investigation that at $\beta>\beta_c$ we find
\begin{equation}
    \lim_{n\to\infty}\mathcal{T}_\alpha^n\mu_\beta^-=\begin{cases}
        \mu_{\beta=\infty}^-&\alpha\in(0,1/4)\\
        \mu_{\beta=0}&\alpha\in(1/4,1/2)
    \end{cases}.
\end{equation}

\usepgfplotslibrary{groupplots}
\pgfplotsset{compat=1.18}

\begin{figure}[t!]
    \begin{tikzpicture}
        \begin{groupplot}[
            group style={
                group size=3 by 1,
                horizontal sep=1.2cm,
            },
            width=0.35\textwidth,
            height=0.35\textwidth,
            xmin=0, xmax=1,
            ymin=0, ymax=1,
            xtick={0, 0.5, 1},
            ytick={0, 0.5, 1},
            grid=major,
            xlabel={$x$},
            samples=100,
            domain=0:1,
            no markers 
        ]
        
        \nextgroupplot[title={$\alpha < 1/4$}, ylabel={$g^\alpha(x)$}]
            \addplot[gray, dashed, thick] {x}; 
            \addplot[blue, ultra thick] {4*0*x + (0.5 - 2*0)*(6*x^2 - 4*x^3)};
    
        \nextgroupplot[title={$\alpha = 1/4$}]
            \addplot[gray, dashed, thick] {x};
            \addplot[red, ultra thick] {4*0.25*x + (0.5 - 2*0.25)*(6*x^2 - 4*x^3)};
    
        \nextgroupplot[title={$\alpha > 1/4$}]
            \addplot[gray, dashed, thick] {x};
            \addplot[green!60!black, ultra thick] {4*0.5*x + (0.5 - 2*0.5)*(6*x^2 - 4*x^3)};
        
        \end{groupplot}
        
    \end{tikzpicture}

    \caption{The three possible dynamics of $g^\alpha$ depending on the value of $\alpha$.}
    \label{Fig4}
\end{figure}
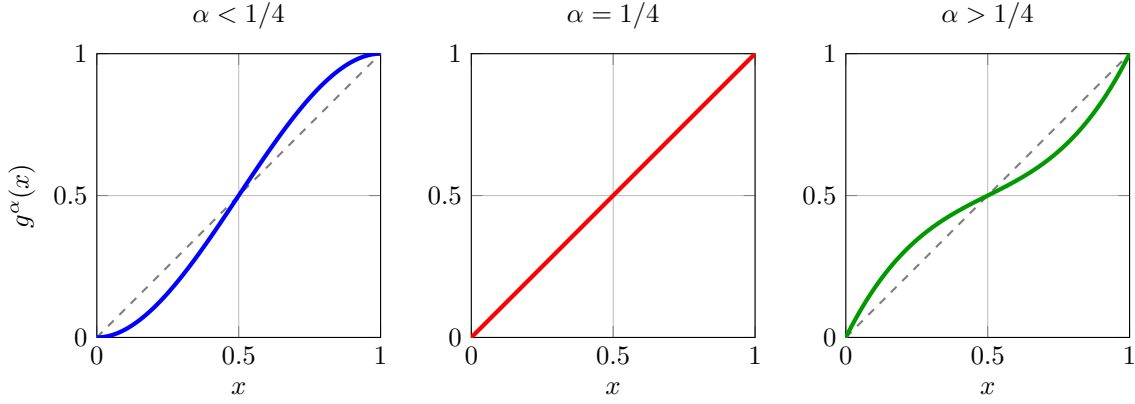

However, we must be careful as the dynamics of $\mu_\beta^-[f_n^\alpha(\textbf{x}^+)]$ only obey the dynamics of $g_\alpha$ asymptotically. For finite $n$, we still get some corrections due to correlations with neighbouring spins. Thus, we must prove that these corrections cannot move the system outside of the basin of attraction of the expected fixed point. For this conclusion, it suffices to prove that $\mu_\beta^-[f_n^\alpha(\textbf{x}^+)]\leq C<1/2\,\forall n\in\mathds{N}$ for $\alpha\in[0,1/4)$ as well as $0<c\leq\mu_\beta^-[f_n^\alpha(\textbf{x}^+)]\leq\frac{1}{2}\,\forall n\in\mathds{N}$ for $\alpha\in(1/4,1/2]$. This we shall do via the following lemma.

\begin{lemma}
    For any $\alpha\in[0,1/4]$ we find
    \begin{equation}
        \lim_{n\to\infty}\mu_\beta^-[f_n^\alpha(\mathrm{\mathbf{x}}^+)]\leq\frac{1-\theta_\beta}{2}
    \end{equation}
    and for any $\alpha\in[1/4,1/2]$ we get
    \begin{equation}
        \frac{1-\theta_\beta}{2}\leq\lim_{n\to\infty}\mu_\beta^-[f_n^\alpha(\mathrm{\mathbf{x}}^+)]\leq\frac{1}{2}.
    \end{equation}
\end{lemma}

\begin{proof}
    Recall that $f_n^{\alpha=1/4}$ is linear and therefore, by translation invariance of the infinite volume Gibbs measure
    \begin{equation}
        \mu_\beta^-[f_n^{\alpha=1/4}(\textbf{x}^+)]=\mu_\beta^-[(\textbf{x}^+)_0]=\frac{\phi_{p_\beta}^1[0\nlr\infty]}{2}=\frac{1-\theta_\beta}{2}\quad\forall n\in\mathds{N}.
    \end{equation}
    For arbitrary $\alpha$, we want to prove by induction that the value remains bounded either below or above the value for $\alpha=1/4$, depending on which region of $\alpha$ we consider. In order to do so, we shall quickly compute
    \begin{equation}
        \frac{\mathrm{d}f^\alpha}{\mathrm{d}\alpha}=\sum_{i}x_i-2\left(\sum_{i<j}x_ix_j-\sum_{i<j<k}x_ix_jx_k\right)=\begin{cases}
            0&\sum_ix_i=0\\
            1&\sum_ix_i=1\\
            0&\sum_ix_i=2\\
            -1&\sum_ix_i=3\\
            0&\sum_ix_i=4
        \end{cases}.
    \end{equation}
    Therefore, for $n=1$ we find that
    \begin{equation}
        \mu_{\mathcal{C}_\Lambda^0,\beta}^-\left[\frac{\mathrm{d}f_1^\alpha}{\mathrm{d}\alpha}\right]=\mu_{\mathcal{C}_\Lambda^0,\beta}^-\left[\sum_i(\textbf{x}^+)_i=1\right]-\mu_{\mathcal{C}_\Lambda^0,\beta}^-\left[\sum_i(\textbf{x}^+)_i=3\right].
    \end{equation}
    We may choose $\Lambda$ as the square $[-m,m]^2$ and consider the spin at the origin. Doing so, we profit from exact rotation symmetry in finite volume already. Thus, we can just choose one specific configuration satisfying $\sum_i(\textbf{x}^+)_i=1$ and one satisfying $\sum_i(\textbf{x}^+)_i=3$ and compare those. Let us therefore compute
    \begin{align}
        \mu_{\mathcal{C}_\Lambda^0,\beta}^-\left[\ind{\sigma_1=\sigma_2=\sigma_3=-1,\sigma_4=+1}\right]-\mu_{\mathcal{C}_\Lambda^0,\beta}^-\left[\ind{\sigma_1=\sigma_2=\sigma_3=+1,\sigma_4=-1}\right]\nonumber\\
        =\mu_{\mathcal{C}_\Lambda^0,\beta}^+\left[\ind{\sigma_1=\sigma_2=\sigma_3=+1,\sigma_4=-1}\right]-\mu_{\mathcal{C}_\Lambda^0,\beta}^-\left[\ind{\sigma_1=\sigma_2=\sigma_3=+1,\sigma_4=-1}\right].
    \end{align}
    To continue, we condition these expectations on the spins $\sigma_1,\sigma_4$. By the law of total expectation we then obtain
    \begin{equation}
        \mu_{\mathcal{C}_\Lambda^0,\beta}^+[\ind{\sigma_1=+1,\sigma_4=-1}\mu_{\mathcal{C}_\Lambda^0,\beta}^+[\ind{\sigma_2=\sigma_3=+1}]\,|\,\sigma_1,\sigma_4]]-\mu_{\mathcal{C}_\Lambda^0,\beta}^-[\ind{\sigma_1=+1,\sigma_4=-1}\mu_{\mathcal{C}_\Lambda^0,\beta}^-[\ind{\sigma_2=\sigma_3=+1}\,|\,\sigma_1,\sigma_4]].
    \end{equation}
    The fact that $\ind{\sigma_2=\sigma_3=+1}$ is increasing together with stochastic domination $\mu_{\mathcal{C}_\Lambda^0,\beta}^-[\cdot\,|\,\sigma_1,\sigma_4]\preceq\mu_{\mathcal{C}_\Lambda^0,\beta}^+[\cdot\,|\,\sigma_1,\sigma_4]$ gives us
    \begin{equation}
        \mu_{\mathcal{C}_\Lambda^0,\beta}^+[\ind{\sigma_2=\sigma_3=+1}]\,|\,\sigma_1,\sigma_4]\geq\mu_{\mathcal{C}_\Lambda^0,\beta}^-[\ind{\sigma_2=\sigma_3=+1}\,|\,\sigma_1,\sigma_4].
    \end{equation}
    Further, by symmetry we also get
    \begin{equation}
        \mu_{\mathcal{C}_\Lambda^0,\beta}^+[\ind{\sigma_1=+1,\sigma_4=-1}]=\mu_{\mathcal{C}_\Lambda^0,\beta}^-[\ind{\sigma_1=+1,\sigma_4=-1}].
    \end{equation}
    All in all, we conclude that
    \begin{equation}
        \mu_{\mathcal{C}_\Lambda^0,\beta}^-\left[\frac{\mathrm{d}f_1^\alpha}{\mathrm{d}\alpha}\right]\geq0,
    \end{equation}
    which implies that $\mu_{\mathcal{C}_\Lambda^0,\beta}^-[f_1^\alpha(\textbf{x}^+)]$ is increasing in $\alpha$. For the induction step, we may assume that 
    \begin{equation}\label{Eq10}
        \mu_{\mathcal{C}_\Lambda^0,\beta}^-[f_k^{\alpha<1/4}(\textbf{x}^+)]\leq\mu_{\mathcal{C}_\Lambda^0,\beta}^-[f_k^{\alpha=1/4}(\textbf{x}^+)]\leq\mu_{\mathcal{C}_\Lambda^0,\beta}^-[f_k^{\alpha>1/4}(\textbf{x}^+)]
    \end{equation}
    holds for all $k\leq n$ for some $n\in\mathds{N}$. We then must prove that this still holds for $k=n+1$. Let us start by considering
    \begin{equation}
        \mu_{\mathcal{C}_\Lambda^0,\beta}^-[f_{n+1}^\alpha(\textbf{x}^+)]=\mathcal{T}_\alpha^n\mu_{\mathcal{C}_\Lambda^0,\beta}^-[f^\alpha(\textbf{x}^+)].
    \end{equation}
    We note that
    \begin{equation}
        \mathcal{T}_\alpha^n\mu_{\mathcal{C}_\Lambda^0,\beta}^-[f^{\alpha=1/4}(\textbf{x}^+)]=\mathcal{T}_\alpha^n\mu_{\mathcal{C}_\Lambda^0,\beta}^-[(\textbf{x}^+)_i]=\mu_{\mathcal{C}_\Lambda^0,\beta}^-[f_n^\alpha(\textbf{x}^+_i)]
    \end{equation}
    for which we can use the induction hypothesis. Thus, the proof reduces itself to the statement
    \begin{equation}
        \mathcal{T}_\alpha^n\mu_{\mathcal{C}_\Lambda^0,\beta}^-\left[\frac{\mathrm{d}f^\alpha}{\mathrm{d}\alpha}\right]\geq0,
    \end{equation}
    which is just the $n=1$ case with $\mu_{\mathcal{C}_\Lambda^0,\beta}^-$ replaced by $\mathcal{T}_\alpha^n\mu_{\mathcal{C}_\Lambda^0,\beta}^-$. Since the kernel $\mathcal{T}_\alpha$ preserves both the symmetries of the Ising model as well as the stochastic ordering (recall condition R3 from the introduction which gives us $\mu_{\mathcal{C}_\Lambda^0,\beta}^-\preceq\mu_{\mathcal{C}_\Lambda^0,\beta}^+\implies \mathcal{T}^n\mu_{\mathcal{C}_\Lambda^0,\beta}^-\preceq \mathcal{T}^n\mu_{\mathcal{C}_\Lambda^0,\beta}^+$), this statement also holds for $\mathcal{T}_\alpha^n\mu_{\mathcal{C}_\Lambda^0,\beta}^-$. We conclude that (\ref{Eq10}) holds for arbitrary $k\in\mathds{N}$ as well as for all $\mathcal{C}_\Lambda^0$ whenever $\Lambda$ is of the form $[-m,m]^2$. In the limit $m\to\infty\implies\Lambda\uparrow\mathds{Z}^2$ we then find
    \begin{equation}
        \mu_\beta^-[f_n^{\alpha<1/4}(\textbf{x}^+)]\leq\mu_\beta^-[f_n^{\alpha=1/4}(\textbf{x}^+)]=\frac{1-\theta_\beta}{2}\leq\mu_\beta^-[f_n^{\alpha>1/4}(\textbf{x}^+)].
    \end{equation}
    At last, by invoking stochastic domination one more time, the upper bound of $1/2$ is trivial as
    \begin{equation}
        1=\mu_\beta^-[f_n^\alpha(\textbf{x}^+)]+\mu_\beta^-[f_n^\alpha(\textbf{x}^-)]=\mu_\beta^-[f_n^\alpha(\textbf{x}^+)]+\mu_\beta^+[f_n^\alpha(\textbf{x}^+)].
    \end{equation}
    For $\alpha\in[0,1/2]$, $f_n^\alpha$ is increasing and thus $\mu_\beta^+[f_n^\alpha(\textbf{x}^+)]\geq\mu_\beta^-[f_n^\alpha(\textbf{x}^+)]\implies \mu_\beta^-[f_n^\alpha(\textbf{x}^+)]\leq1/2$.
\end{proof}

The remaining gap towards establishing the exact picture we presented in the introduction is the scaling limit of $\mu_{\beta_c}$ in the RG-regime $\alpha>1/4$. In order to prove this, we shall once more make great use of stochastic domination.

\begin{proposition}\label{Cor1}
    For $\alpha\in[1/4,1/2)$ we find
    \begin{equation}
        \lim_{n\to\infty}\mathcal{T}_\alpha^n\mu_{\beta_c}=\mu_{\beta=0}.
    \end{equation}
\end{proposition}

\begin{proof}
    We have already established this result for $\alpha=1/4$. However, we can now extend it to $\alpha\in(1/4,1/2)$ using a variant of the sandwich lemma. Following the monotone coupling idea presented in \cite{friedli_velenik_2017} in Chapter 3.10.3, we find
    \begin{equation}
        \mu_G(\sigma_x=+1|\sigma_{G\backslash\{x\}})\leq\nu_G(\sigma_x=+1|\sigma_{G\backslash\{x\}})\,\,\forall x\in G,\forall \sigma\in\Omega_G\implies\mu_G\preceq\nu_G,
    \end{equation}
    where $G\subset\mathds{Z}^2$ is some finite subgraph. For the Ising model, where $\mu_{G,\beta_c}^\mathrm{f}(\sigma_x=+1|\sigma_{G\backslash\{x\}})$ only depends on the state of the nearest-neighbours, computing these single-site conditional probabilities is a straightforward exercise, giving us
    \begin{equation}
        \frac{1}{1+e^{8\beta_c}}\leq\mu_{G,\beta_c}^\mathrm{f}(\sigma_x=+1|\sigma_{G\backslash\{x\}})\leq\frac{1}{1+e^{-8\beta_c}}
    \end{equation}
    Hence, there exist product measures $\pi_{G,p},\pi_{G,p'}$ defined as
    \begin{equation}\label{Eq11}
        \pi_{G,q}=\prod_{x\in G}(q\delta_++(1-q)\delta_-),q=p,p'
    \end{equation}
    with values $p=\frac{1}{1+e^{8\beta_c}},p'=\frac{1}{1+e^{-8\beta_c}}$ satisfying
    \begin{equation}
        \pi_{G,p}\preceq\mu_{G,\beta_c}^\mathrm{f}\preceq\pi_{G,p'}.
    \end{equation}
    Since this holds for arbitrary $G\subset\mathds{Z}^2$, we obtain $\pi_p\preceq\mu_{\beta_c}\preceq\pi_{p'}$ in the infinite volume limit, where $\pi_p,\pi_{p'}$ are simply the same products as in (\ref{Eq11}) taken over all of $\mathds{Z}^2$. Due to the product structure already apparent in the input measures, computing the limit $\lim_{n\to\infty}\mathcal{T}_\alpha^n\pi_q,q=p,p'$ follows exactly the dynamics of $g^\alpha(x)$ displayed in Figure \ref{Fig4}. Hence, we conclude that for $\alpha\in(1/4,1/2]$
    \begin{equation}
        \lim_{n\to\infty}\mathcal{T}_\alpha^n\pi_q=\pi_{1/2}=\mu_{\beta=0}\quad\forall q\in(0,1).
    \end{equation}
    This proves the desired statement via the sandwich
    \begin{equation}
        \mu_{\beta=0}=\lim_{n\to\infty}\mathcal{T}_\alpha^n\pi_p\preceq\lim_{n\to\infty}\mathcal{T}_\alpha^n\mu_{\beta_c}\preceq\lim_{n\to\infty}\mathcal{T}_\alpha^n\pi_{p'}=\mu_{\beta=0}\implies\lim_{n\to\infty}\mathcal{T}_\alpha^n\mu_{\beta_c}=\mu_{\beta=0}.
    \end{equation}
\end{proof}

This completes the proof of the main result stated in the introduction. The limit of the RG-flow for the $2\times 2\to1$ coarse-graining is given by
\begin{align}
    \alpha\in(0,1/4)&:\quad \lim_{n\to\infty}\mathcal{T}_\alpha^n\mu_\beta^\pm=\begin{cases}
            \mu_{\beta=0}&\beta<\beta_c\\
            \,\,\,\,\,?&\beta=\beta_c\\
            \mu_{\beta=\infty}^\pm&\beta>\beta_c
        \end{cases},\nonumber\\
        \alpha=1/4&:\quad\lim_{n\to\infty}\mathcal{T}_\alpha^n\mu_\beta^\pm=\prod_{x\in\mathds{Z}^2}\left(\frac{1+m_\beta}{2}\delta_{\pm1}(\sigma_x)+\frac{1-m_\beta}{2}\delta_{\mp1}(\sigma_x)\right),\\
        \alpha\in(1/4,1/2)&:\quad\lim_{n\to\infty}\mathcal{T}_\alpha^n\mu_\beta^\pm=\begin{cases}
            \mu_{\beta=0}&\beta<\infty\\
            \mu_{\beta=\infty}^\pm&\beta=\infty
        \end{cases}.\nonumber
\end{align}
Although not formally proven, it seems very believable that these results simply extend to the cases $\alpha=0,1/2$. Thus, the remaining open question is what happens in the case $\alpha\in[0,1/4)$ and $\beta=\beta_c$. It is hoped that in this scenario, the limit really is given by a non-trivial fixed point. A full proof, however, will require some more sophisticated methods and arguments. Nevertheless, we have gained at least some insight to the RG-flow at criticality. Recalling the argument presented in the proof of Proposition \ref{Cor1}, it is clear that the RG-flow must exhibit some sort of saddle-point behaviour in order for there to be hope for the existence of a non-trivial fixed point at $\beta_c$.

\section{Generalizations}\label{Sec5}

There are two generalizations we shall shortly discuss before we conclude this article:
\begin{enumerate}
    \item[i)] higher dimensions,
    \item[ii)] different single-spin spaces like higher-state Potts-models.
\end{enumerate}
For i) let us consider a general RGT on $\mathds{Z}^3$ using cubic blocks. Besides R1-R4 we also want to simplify our lives a bit by imposing that the single-block kernel $t(\sigma|s)$ is a function depending only on $\sum_xs_x$ and not the geometric positioning of the spins\footnote{This was already the case in $2\mathrm{D}$, although there it was a consequence of the spatial symmetries and not something that had to be imposed.}. It is then also clear that the cluster weights of connecting the coarse spin to $n$ fine spins only depend on $n$ and not on the geometric choices of edges. Let us denote the cluster weights of $n=1,...,8$ open edges as $A,...,H$. In that case, we obtain a three parameter family of RGTs given by\footnote{One should not get confused by the double usage of $\beta$, here it simply refers to a paramter in the RG and not the inverse temperature of a model.}
\begin{align}
    8+\to+&=1=8A+28B+56C+70D+56E+28F+8G+H,\nonumber\\
    7+\to+&=1-\alpha=7A+21B+35C+35D+21E+7F+G,\nonumber\\
    6+\to+&=1-\beta=6A+15B+20C+15D+6E+F,\nonumber\\
    5+\to+&=1-\gamma=5A+10B+10C+5D+E,\nonumber\\
    4+\to+&=1/2=4A+6B+4C+D,\\
    3+\to+&=\gamma=3A+3B+C,\nonumber\\
    2+\to +&=\beta=2A+B,\nonumber\\
    1+\to+&=\alpha=A,\nonumber\\
    0+\to+&=0,\nonumber
\end{align}

for the allowed values $0\leq\alpha\leq\beta\leq\gamma\leq1/2$. Solving this system of equations gives us the cluster weights
\begin{gather*}
    A = \alpha,\qquad
    B = \beta - 2\alpha,\qquad
    C = \gamma - 3\beta + 3\alpha,\qquad
    D = -4\alpha + 6\beta - 4\gamma + \tfrac{1}{2}, \\
    E = 5\alpha - 10\beta + 9\gamma - \tfrac{3}{2},\qquad
    F = -6\alpha + 14\beta - 14\gamma + \tfrac{5}{2},\qquad
    G = 6\alpha - 14\beta + 14\gamma - \tfrac{5}{2},\qquad
    H = 0.
\end{gather*}
The linear map, corresponding to $A=1/8,B=C=\cdots =H=0$ is given by $\alpha=1/8,\beta=2/8$ and $\gamma=3/8$. These weights then allow us to compute the evolution map $f$ which is a polynomial of degree seven with coefficients $A-G$. Computing the derivative of $f$ gives us
\begin{equation}
    \delta f=\max\{\alpha,\beta-\alpha,\gamma-\beta,1/2-\gamma\}.
\end{equation}
This has a nice interpretation as it is the maximal width of any interval in the partition $[0,\alpha,\beta,\gamma,1/2]$. Now, the proof for the high-temperature phase only used the exponential decay of clusters which is established in all dimensions. Thus, for any value of $\alpha,\beta,\gamma$ that satisfy $8\cdot(\delta f)^2<1$ the proof still holds and we obtain convergence towards the $T=\infty$ fixed point within all of the high-temperature phase.

For the discussion of the low-temperature phase we are interested in the dynamics of $f$. For all values $0\leq\alpha\leq\beta\leq\gamma$ such that they are also strictly below their respective values in the linear case $\alpha<1/8,\beta<2/8$ and $\gamma<3/8$, the dynamics of $f$ look equivalent to the two-dimensional case at $\alpha_{2\mathrm{D}}<1/4$. However, for values beyond the linear case there are certain triplets that produce more complicated dynamics than what we had seen in the $2\mathrm{D}$ case. Explicitly, there are some choices that lead to five rather than just three fixed points:

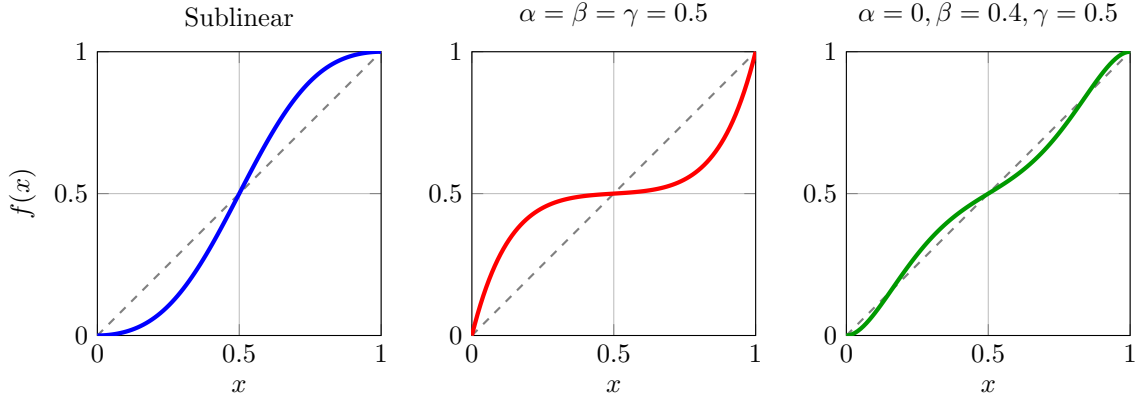
\begin{figure}[h!]
    \centering
    \begin{tikzpicture}
        \begin{groupplot}[
            group style={
                group size=3 by 1,
                horizontal sep=1.2cm,
            },
            width=0.35\textwidth,
            height=0.35\textwidth,
            xmin=0, xmax=1,
            ymin=0, ymax=1,
            xtick={0, 0.5, 1},
            ytick={0, 0.5, 1},
            grid=major,
            xlabel={$x$},
            samples=100,
            domain=0:1,
            no markers 
        ]
        
        \nextgroupplot[title={Sublinear}, ylabel={$f(x)$}]
            \addplot[gray, dashed, thick] {x}; 
            \addplot[blue, ultra thick] {28*0.05*x^2 + 56*(0.1-3*0.05)*x^3 + 70*(0.5+6*0.05-4*0.1)*x^4 + 56*(9*0.1-10*0.05-1.5)*x^5 +28*(2.5+14*0.05-14*0.1)*x^6+8*(14*0.1-14*0.05-2.5)*x^7};
    
        \nextgroupplot[title={$\alpha=\beta=\gamma=0.5$}]
            \addplot[gray, dashed, thick] {x};
            \addplot[red, ultra thick] {8*0.5*x + 28*(0.5-2*0.5)*x^2 + 56*(3*0.5+0.5-3*0.5)*x^3 + 70*(0.5-4*0.5+6*0.5-4*0.5)*x^4 + 56*(5*0.5+9*0.5-10*0.5-1.5)*x^5 +28*(2.5-6*0.5+14*0.5-14*0.5)*x^6+8*(6*0.5+14*0.5-14*0.5-2.5)*x^7};
    
        \nextgroupplot[title={$\alpha=0,\beta=0.4,\gamma=0.5$}]
            \addplot[gray, dashed, thick] {x};
            \addplot[green!60!black, ultra thick] {28*0.4*x^2 + 56*(0.5-3*0.4)*x^3 + 70*(0.5+6*0.4-4*0.5)*x^4 + 56*(9*0.5-10*0.4-1.5)*x^5 +28*(2.5+14*0.4-14*0.5)*x^6+8*(14*0.5-14*0.4-2.5)*x^7};
        
        \end{groupplot}
    \end{tikzpicture}
    \caption{The dynamics of $f$ for different values of the input parameters.}
    \label{fig:placeholder}
\end{figure}

\newpage

However, in order to conclude the "correct" scaling limit in the low-temperature phase we would also have to prove that the variance of $f_n$ vanishes. In our proof for the Ising model on $\mathds{Z}^2$ we used the concept of duality which is much stronger in two dimensions compared to higher ones. Although, the actual statement we proved using duality was that conditioning on $\omega_e=0,1$ has exponential decaying impact on the cluster configurations farther away. Establishing such a result in higher dimensions should not be too much of an obstacle. It remains to be checked whether one can make it work for this specific application.\\

Now, let us also investigate our second proposed generalization. Surprisingly, once we turn our attention to higher-state Potts-models, we no longer obtain any RGT besides the linear one that is representable as a cluster model in the simplest form we had used so far. The following proposition presents a proof of this fact.

\begin{proposition}\label{Prop3}
    Let $B$ be a block of size $|B|>2$. Then there exists no RGT for Potts models with $q\geq3$ that is compatible with cluster weights besides the linear transformation obtained by forbidding more than one edge to be open simultaneously.
\end{proposition}

\begin{proof}
    We want to prove inductively, that all edge weights of edge-configurations with more than one open edge must be zero. Let us denote the states as $a,b,c,...$. Further, let $A_2$ be (the weight of) some edge configuration that connects the coarse spin to exactly two fine spins. Further, let us consider configurations for which all fine spins not hit by $A_2$ are in state $a$. The remaining two states can be either $b,b$ or $b,c$. Due to the normalization condition we find that
    \begin{equation}
        \left[(a,...,a,b,b)\to b\right]=1-\left[(a,...,a,b,b)\to a\right]=\left[(a,...,a,b,c)\to b\right]+\left[(a,...,a,b,c)\to c\right].
    \end{equation}
    The cluster weight of $[(a,...,a,x,y)\to a]$ is independent of the states $x,y$. Thus, we now find that
    \begin{equation}
        2A_1+A_2=2A_1\implies A_2=0.
    \end{equation}
    Therefore, cluster weights for all edge configurations that hit two coarse spins are zero. We can extend this inductively to any number of edges by considering transition probabilities $(a,...,a,b,...,b,c)\to b/c$ and $(a,...,a,b,...,b)\to b$. We conclude that the only possible cluster representable RGT is the one given by $A_n=0,n>1$. Note that we did not have to assume that all edge configurations with $n$ open edges have the same weight. This result even holds if we allow the cluster weights to depend on the geometric positioning of the edges.
\end{proof}

The reason this does not fail for the Ising model is that once we specify all $+$ spins, the complement is also specified as all the remaining spins must be $-$. However, as soon as we have more than two states, specifying all spins of state $a$ leaves some freedom in the choice of states for the remaining spins. This additional information cannot be represented by vertical bonds alone. There is away to circumvent this issue. Let us highlight this idea on the example of the $3$-states Potts model on the $2\mathrm{D}$ triangular-lattice using a $3\to1$ coarse-graining with triangular blocks. The possible single block-configurations we must consider are the following:

\begin{figure}[htpb!]
    \centering
    \begin{tikzpicture}
        
        \draw[black][-] (-0.85,0) to (0.85,0);
        \draw[black][-] (-0.933,0.134) to (-0.067,1.596);
        \draw[black][-] (0.933,0.134) to (0.067,1.596);
        \draw[black] (-1,0) circle (0.15);
        \draw[black] (1,0) circle (0.15);
        \draw[black] (0,1.73) circle (0.15);
        \node[] at (0,-0.35) {$1$};
        \node[] at (-1,0) {\scriptsize$a$};
        \node[] at (1,0) {\scriptsize$a$};
        \node[] at (0,1.73) {\scriptsize$a$};
        \node[] at (0,0.577) {\scriptsize$a$};

        \begin{scope}[shift={(3.5,0)}]

            \draw[black][-] (-0.85,0) to (0.85,0);
            \draw[black][-] (-0.933,0.134) to (-0.067,1.596);
            \draw[black][-] (0.933,0.134) to (0.067,1.596);
            \draw[black] (-1,0) circle (0.15);
            \draw[black] (1,0) circle (0.15);
            \draw[black] (0,1.73) circle (0.15);
            \node[] at (0,-0.35) {$1-\alpha$};
            \node[] at (-1,0) {\scriptsize$a$};
            \node[] at (1,0) {\scriptsize$a$};
            \node[] at (0,1.73) {\scriptsize$b$};
            \node[] at (0,0.577) {\scriptsize$a$};
        \end{scope}

        \begin{scope}[shift={(7,0)}]
            \draw[black][-] (-0.85,0) to (0.85,0);
            \draw[black][-] (-0.933,0.134) to (-0.067,1.596);
            \draw[black][-] (0.933,0.134) to (0.067,1.596);
            \draw[black] (-1,0) circle (0.15);
            \draw[black] (1,0) circle (0.15);
            \draw[black] (0,1.73) circle (0.15);
            \node[] at (0,-0.35) {$\alpha$};
            \node[] at (-1,0) {\scriptsize$a$};
            \node[] at (1,0) {\scriptsize$a$};
            \node[] at (0,1.73) {\scriptsize$b$};
            \node[] at (0,0.577) {\scriptsize$b$};
        \end{scope}

        \begin{scope}[shift={(10.5,0)}]
            \draw[black][-] (-0.85,0) to (0.85,0);
            \draw[black][-] (-0.933,0.134) to (-0.067,1.596);
            \draw[black][-] (0.933,0.134) to (0.067,1.596);
            \draw[black] (-1,0) circle (0.15);
            \draw[black] (1,0) circle (0.15);
            \draw[black] (0,1.73) circle (0.15);
            \node[] at (0,-0.35) {$1/3$};
            \node[] at (-1,0) {\scriptsize$a$};
            \node[] at (1,0) {\scriptsize$b$};
            \node[] at (0,1.73) {\scriptsize$c$};
            \node[] at (0,0.577) {\scriptsize$a$};
        \end{scope}

    \end{tikzpicture}
    \caption{Possible block configurations for a $3\to1$ coarse-graining of the $3$-states Potts model.}
    \label{Fig3}
\end{figure}
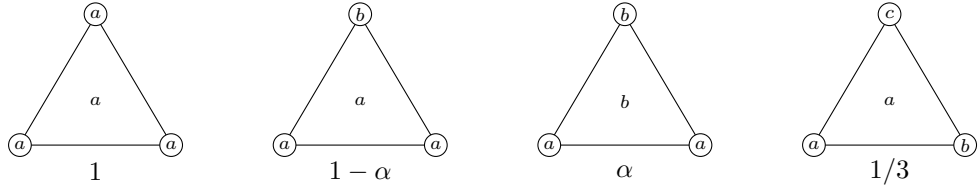

Again, we want to set the probability of the coarse spin being anti-parallel to all fine spins to zero. If we restricted ourselves to vertical bonds, Proposition \ref{Prop3} would imply that $\alpha=1/3$ is the only RGT in this family that allows for a cluster representation. However, we can also allow horizontal bonds between the fine spins. This then gives us the following bond configurations:

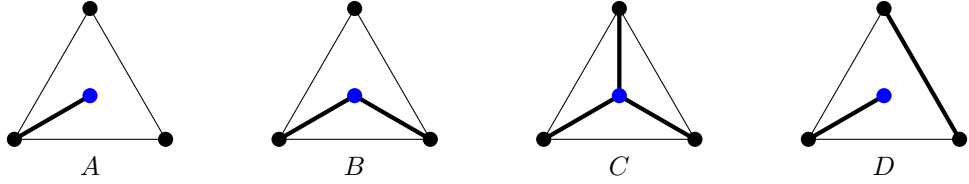
\begin{figure}[htpb!]
    \centering
    \begin{tikzpicture}
        
        \draw[black][ultra thick][-] (-1,0) to (0,0.577);
        \draw[black][ultra thick][-] (1,0) to (0,0.577);
        \draw[black][ultra thick][-] (0,1.73) to (0,0.577);
        \draw[black][-] (-1,0) to (1,0);
        \draw[black][-] (-1,0) to (0,1.73);
        \draw[black][-] (1,0) to (0,1.73);
        \fill[black] (-1,0) circle (0.1);
        \fill[black] (1,0) circle (0.1);
        \fill[black] (0,1.73) circle (0.1);
        \fill[blue] (0,0.577) circle (0.1);
        \node[] at (0,-0.35) {$C$};

        \begin{scope}[shift={(-3.5,0)}]
            \draw[black][ultra thick][-] (-1,0) to (0,0.577);
            \draw[black][ultra thick][-] (1,0) to (0,0.577);
            \draw[black][-] (-1,0) to (1,0);
            \draw[black][-] (-1,0) to (0,1.73);
            \draw[black][-] (1,0) to (0,1.73);
            \fill[black] (-1,0) circle (0.1);
            \fill[black] (1,0) circle (0.1);
            \fill[black] (0,1.73) circle (0.1);
            \fill[blue] (0,0.577) circle (0.1);
            \node[] at (0,-0.35) {$B$};
        \end{scope}

        \begin{scope}[shift={(-7,0)}]
            \draw[black][ultra thick][-] (-1,0) to (0,0.577);
            \draw[black][-] (-1,0) to (1,0);
            \draw[black][-] (-1,0) to (0,1.73);
            \draw[black][-] (1,0) to (0,1.73);
            \fill[black] (-1,0) circle (0.1);
            \fill[black] (1,0) circle (0.1);
            \fill[black] (0,1.73) circle (0.1);
            \fill[blue] (0,0.577) circle (0.1);
            \node[] at (0,-0.35) {$A$};
        \end{scope}

        \begin{scope}[shift={(3.5,0)}]
            \draw[black][ultra thick][-] (-1,0) to (0,0.577);
            \draw[black][-] (-1,0) to (1,0);
            \draw[black][-] (-1,0) to (0,1.73);
            \draw[black][ultra thick][-] (1,0) to (0,1.73);
            \fill[black] (-1,0) circle (0.1);
            \fill[black] (1,0) circle (0.1);
            \fill[black] (0,1.73) circle (0.1);
            \fill[blue] (0,0.577) circle (0.1);
            \node[] at (0,-0.35) {$D$};
        \end{scope}
    \end{tikzpicture}
    \caption{Bond configurations including horizontal bonds.}
    \label{fig:placeholder}
\end{figure}

The last configuration is the only additional one we have to consider. First of all, because we demand that configurations where the coarse spin is anti-parallel to all fine spins have a probability of zero, automatically all bond-configurations where the coarse spin is not connected to any fine spin have a weight of zero as well. Hence, we must not consider bond configurations which consist of only horizontal bonds. On top of that if we take $D$ and add any other open edge then the connectivity becomes the same as in $C$ such that we have already included these cases. By matching these cluster weights with the transition probabilities from Figure \ref{Fig3} we obtain the following equations.
\begin{equation}
    \begin{aligned}
        3A+3B+C+3D&=1\\
    2A+B&=1-\alpha\\
    A+D&=\alpha\\
    A&=1/3
    \end{aligned}\quad\implies\quad\begin{aligned}
        A&=1/3\\
        B&=1/3-\alpha\\
        C&=0\\
        D&=\alpha-1/3
    \end{aligned}
\end{equation}
While this still forces $A=1/3$, the introduction of the $D$ configuration allows us to represent the RGT for arbitrary values of $\alpha$. Thus, we have found a valid way to represent such a transformation with cluster rules. A key thing to note is that although there suddenly can be horizontal bonds in layers above the base level, these horizontal bonds are not as bad as they seem. They are absent from the topmost layer after $n$ RG-steps and they are restricted to edges within the RG-block of a given spin at the top. Therefore, connections between spins at the top can still only come through connections from the input model at the base layer. Nevertheless, it remains an open question whether the arguments we have used for the Ising model also apply in such a construction or not.

\section{Conclusions}

So far, the advantages of the cluster representation have definitely not been fully explored, although it is questionable whether the negative cluster weights will ever allow for connectivity arguments to be applied at full strength to RGTs. Yet this might be an avenue that is worth exploring further regarding the question of the RG-flow at criticality. It is intriguing to see whether powerful arguments from the world of cluster representations like crossing probabilities can be combined with cluster representations of RGTs. So far, we had always kept the spins in the bottom layer and only worked with clusters in the RG-regime, i.e. the layers above. The obvious next step to tackle the RG-flow at criticality is to represent the Ising spins in the bottom layer via the Edwards-Sokal coupling and then integrate out all spins, thus only keeping bond variables on the entirety of the lattice. Doing so, one can represent correlations between coarse spins by signed weights of connectivity events. Two coarse spins are connected if and only if they connect to a common Ising cluster via vertical bonds. The main difficulty in this approach are the negative cluster weights which forces one to keep analytic control over the sign at all steps. We shall see whether in the future this approach leads to a complete argument or not.\medskip

The concept of the evolution map itself has already been sufficient to obtain a complete picture of the RG-flow away from criticality for almost the full range of physically reasonable RGTs on the $2\times 2\to1$ coarse-graining geometry on $\mathds{Z}^2$. Especially the convergence of the entire low-temperature phase to the trivial $T=0$ fixed point for $\alpha<1/4$ seems to be a novel result. The generalization of these results to higher dimensions appears as a straightforward exercise. Yet the most striking features of RG-theory remain outside of what seems possible for now. Among them are the existence of a non-trivial fixed point as well as universality (keep in mind that, while the authors of \cite{Chelkak-Smirnov} have proven universality among the Ising models on different lattices, we are more so thinking about universality among different models, i.e.\ $\{-1,1\}$-spin models with arbitrary couplings). The arguments presented in this paper are tailored for the nearest-neighbour Ising model and, therefore, make no statement on the RG-flow of systems with arbitrary couplings.

\section*{Acknowledgements}

The author would like to thank his supervisor Uwe-Jens Wiese for the initial inspiration with this idea as well as continuous support throughout the entirety of this project and beyond. Further gratitude is due to Matthias Blau, Sebastian Baader, Urs Wenger and Nico Scheidegger for fruitful discussions. A special thank you is also due to Slava Rychkov and Stanislav Smirnov for reading a previous version of this manuscript and providing helpful comments.

\printbibliography

@book{friedli_velenik_2017,
place={Cambridge},
title={Statistical Mechanics of Lattice Systems: A Concrete Mathematical Introduction},
DOI={10.1017/9781316882603},
ISBN={978-1-107-18482-4},
publisher={Cambridge University Press},
author={Friedli, Sacha and Velenik, Yvan},
year={2017}
}

@article{van_Enter_1993,
   title={Regularity properties and pathologies of position-space renormalization-group transformations: Scope and limitations of Gibbsian theory},
   volume={72},
   ISSN={1572-9613},
   url={http://dx.doi.org/10.1007/BF01048183},
   DOI={10.1007/bf01048183},
   number={5–6},
   journal={Journal of Statistical Physics},
   publisher={Springer Science and Business Media LLC},
   author={van Enter, Aernout C. D. and Fernández, Roberto and Sokal, Alan D.},
   year={1993},
   month=sep, pages={879–1167} }

@article{Fortuin:1971dw,
    author = "Fortuin, C. M. and Kasteleyn, P. W.",
    title = "{On the Random cluster model. 1. Introduction and relation to other models}",
    doi = "10.1016/0031-8914(72)90045-6",
    journal = "Physica",
    volume = "57",
    pages = "536--564",
    year = "1972"
}

@InProceedings{10.1007/978-3-030-32011-9_2,
author="Duminil-Copin, Hugo",
editor="Barlow, Martin T.
and Slade, Gordon",
title="Lectures on the Ising and Potts Models on the Hypercubic Lattice",
booktitle="Random Graphs, Phase Transitions, and the Gaussian Free Field",
year="2020",
publisher="Springer International Publishing",
address="Cham",
pages="35--161",
isbn="978-3-030-32011-9"
}

@article {CoVe2012,
    AUTHOR = {Coquille, Loren and Velenik, Yvan},
     TITLE = {A finite-volume version of {A}izenman-{H}iguchi theorem for
              the 2d {I}sing model},
   JOURNAL = {Probab. Theory Related Fields},
  FJOURNAL = {Probability Theory and Related Fields},
    VOLUME = {153},
      YEAR = {2012},
    NUMBER = {1-2},
     PAGES = {25--44},
       DOI = {10.1007/s00440-011-0339-6},
       URL = {http://dx.doi.org/10.1007/s00440-011-0339-6},
}

@article{EfronStein,
author = {B. Efron and C. Stein},
title = {{The Jackknife Estimate of Variance}},
volume = {9},
journal = {The Annals of Statistics},
number = {3},
publisher = {Institute of Mathematical Statistics},
pages = {586 -- 596},
year = {1981},
doi = {10.1214/aos/1176345462},
URL = {https://doi.org/10.1214/aos/1176345462}
}

@article{Jensen,
author = {J. L. W. V. Jensen},
title = {{Sur les fonctions convexes et les inégalités entre les valeurs moyennes}},
volume = {30},
journal = {Acta Mathematica},
number = {none},
publisher = {Institut Mittag-Leffler},
pages = {175 -- 193},
year = {1906},
doi = {10.1007/BF02418571},
URL = {https://doi.org/10.1007/BF02418571}
}

@article{Chelkak-Smirnov,
author = {Chelkak, Dmitry and Smirnov, Stanislav},
year = {2009},
month = {10},
pages = {},
title = {Universality in the 2D Ising model and conformal invariance of fermionic
observables},
volume = {189},
journal = {Inventiones mathematicae},
doi = {10.1007/s00222-011-0371-2}
}

@article{Duminil-Copin,
author = {Duminil-Copin, Hugo},
year = {2020},
month = {03},
pages = {},
title = {Exponential Decay of Truncated Correlations for the Ising Model in any Dimension for all but the Critical Temperature
},
volume = {374},
journal = {Commun. Math. Phys.},
doi = {10.1007/s00220-019-03633-y}
}

@article{Kadanoff,
  title = {Scaling laws for Ising models near ${T}_{c}$},
  author = {Kadanoff, Leo P.},
  journal = {Physics Physique Fizika},
  volume = {2},
  issue = {6},
  pages = {263--272},
  numpages = {10},
  year = {1966},
  month = {Jun},
  publisher = {American Physical Society},
  doi = {10.1103/PhysicsPhysiqueFizika.2.263},
  url = {https://link.aps.org/doi/10.1103/PhysicsPhysiqueFizika.2.263}
}

@article{WilsonI,
  title = {Renormalization Group and Critical Phenomena. I. Renormalization Group and the Kadanoff Scaling Picture},
  author = {Wilson, Kenneth G.},
  journal = {Phys. Rev. B},
  volume = {4},
  issue = {9},
  pages = {3174--3183},
  numpages = {0},
  year = {1971},
  month = {Nov},
  publisher = {American Physical Society},
  doi = {10.1103/PhysRevB.4.3174},
  url = {https://link.aps.org/doi/10.1103/PhysRevB.4.3174}
}

@article{WilsonII,
  title = {Renormalization Group and Critical Phenomena. II. Phase-Space Cell Analysis of Critical Behavior},
  author = {Wilson, Kenneth G.},
  journal = {Phys. Rev. B},
  volume = {4},
  issue = {9},
  pages = {3184--3205},
  numpages = {0},
  year = {1971},
  month = {Nov},
  publisher = {American Physical Society},
  doi = {10.1103/PhysRevB.4.3184},
  url = {https://link.aps.org/doi/10.1103/PhysRevB.4.3184}
}

@article{Ising,
    author = {Ising, Ernst},
    title = {Beitrag zur Theorie des Ferromagnetismus},
    journal = {Zeitschrift für Physik},
    year = {1925},
    month = {02},
    doi = {10.1007/BF02980577},
    url = {https://doi.org/10.1007/BF02980577},
    pages = {253--258},
    volume = {31},
    issue = {1},
}

\nocite{*}

\end{document}